\documentclass[12pt,reqno]{amsart}
\usepackage{lmodern}
\usepackage[T1]{fontenc}

\usepackage[numbers]{natbib}
\usepackage{amsmath,amssymb,latexsym} 
\usepackage{fullpage}
\usepackage{etoolbox}
\usepackage{enumitem}
\usepackage{color}
\usepackage[colorlinks,linkcolor=blue,citecolor=magenta]{hyperref}
\usepackage[colorinlistoftodos,bordercolor=orange,backgroundcolor=orange!20,linecolor=orange,textsize=scriptsize]{todonotes}
\usepackage{soul}
\usepackage{microtype}
\usepackage{yhmath}
\usepackage{float}
\usepackage{algorithm}
\usepackage[noend]{algpseudocode}
\usepackage{booktabs}
\usepackage{comment}
\usepackage{mathtools}
\usepackage{calc}
\usepackage{float}
\usepackage{pstricks}
\usepackage{pstricks-add}
\usepackage[flushleft]{threeparttable}
\usepackage{rotating}
\usepackage{graphicx} 
\usepackage{pgfplots}
\usepackage{caption}
\usepackage{multirow}

\usepackage{geometry}
\newtheorem{theorem}{Theorem}
\newtheorem{corollary}{Corollary}

\newtheorem{lemma}{Lemma}

\theoremstyle{definition}

\DeclareFontFamily{U}{mathx}{}
\DeclareFontShape{U}{mathx}{m}{n}{ <-> mathx10 }{}
\DeclareSymbolFont{mathx}{U}{mathx}{m}{n}
\DeclareFontSubstitution{U}{mathx}{m}{n}
\DeclareMathAccent{\widecheck}{0}{mathx}{"71}

\newlength{\LPlhbox}

\newcommand{\mFA}{\text{F1}}
\newcommand{\mFB}{\text{F2}}
\newcommand{\mFC}{\text{F3}}
\newcommand{\mFD}{\text{F4}}

\title{Minimizing the makespan in job shop scheduling under conflict graph constraints}

\author{Nour ElHouda Tellache}
\address{N.E.H. Tellache, Decision Support \& Operations Research Group, Department of Informatics, University of Fribourg, Fribourg, Switzerland}
\email{nourelhouda.tellache@unifr.ch}

\author{Abdenour Azerine}
\address{A. Azerine, Université de Haute-Alsace F-68100 Mulhouse, France}
\email{abdennour.azerine@uha.fr}


\begin{document}
	
	\maketitle
	
	\begin{abstract}
	We study the job shop scheduling problem with a conflict
	graph, in which adjacent jobs cannot be processed
	simultaneously on different machines, with the objective
	of minimizing the makespan. We establish a polynomial
	equivalence with a resource-constrained job shop problem
	having unit-capacity resources and binary resource
	requirements The problem
	is strongly NP-hard even on two machines with unit-time
	operations. For this setting, we derive an optimal-makespan
	formula and a linear-time algorithm when the conflict graph
	is a cluster graph, given the clique partition. We also
	develop a polynomial-time $4/3$-approximation algorithm
	for arbitrary conflict graphs on two machines when each
	job has exactly two unit-time operations. For the general problem, we propose four
	mixed-integer linear formulations, lower bounds, and a
	genetic algorithm with active, non-delay, and hybrid
	decoding procedures. Computational experiments on instances
	derived from the Lawrence and Taillard benchmarks and on
	randomly generated generalized job shop instances evaluate
	the proposed methods and show that the genetic algorithm
	generally improves on constructive heuristics.
	\end{abstract}
	
	\section{Introduction}
	The job shop scheduling problem is a fundamental problem in machine scheduling and ranks among the hardest combinatorial optimization problems. It involves assigning jobs to machines, where each job consists of an ordered sequence of operations, each requiring a specific machine. The goal is to optimize one or more objective functions while adhering to various constraints~\cite{MP08}. This problem has numerous applications in industries such as semiconductor manufacturing, machine manufacturing, automobile production, supply chains, and metallurgy (see~\cite{xiong2022survey} for more details). Each application imposes different requirements, which translate into specific scheduling constraints. In this paper, we focus on a particular class of constraints known as conflict constraints, which naturally arise when certain jobs cannot be processed simultaneously—for instance, when they share resources. We study the problem where these constraints are represented by a simple undirected graph, referred to as the conflict graph.
	
	Formally, the Job Shop scheduling problem with Conflict graph (JSC) consists of scheduling a set of jobs $J=\{J_1,\ldots,J_n\}$ on $m$ dedicated machines $M=\{M_1,\ldots,M_m\}$. Each job $J_j$ comprises $m_j$ operations $O_{ij}$ ($i=1,\ldots,m_j$), such that $O_{ij}$ has to be processed on a specified machine $\mu_{ij} \in M$ for $p_{ij} \in \mathbb{N}^*$ time units. Each job $J_j$ has a predetermined processing order through the machines $O_{1j} \rightarrow O_{2j} \rightarrow \ldots \rightarrow O_{m_jj}$. We assume that any two consecutive operations of the same job are to be performed on different machines. Furthermore, each machine can handle at most one operation at a time and each job is processed on at most one machine at a time. The processing of the jobs is subject to conflict constraints given by a simple undirected graph $G=(V,E)$ called the conflict graph. Each vertex represents a job and two jobs that are adjacent in $G$ are in conflict, i.e. they cannot be processed at the same time on different machines. Likewise, we define conflicts between two operations: two operations $O_{ij}$ and $O_{uv}$ are in conflict (denoted $O_{ij}\sim O_{uv}$)  if they belong to the same job, require the same machine, or are of two conflicting jobs. The objective is to find a feasible schedule that minimizes the maximum completion time also called the makespan and denoted $C_{max}$. According to the three field classification $ \alpha|\beta|\gamma $ of  \citet{GLLR79}, we denote our scheduling problem by $ Jm|ConfG = (V,E)|C_{max} $, where $ ConfG = (V,E) $ indicates the presence of a conflict graph $ G = (V,E) $ over the jobs.
	
	The JSC problem is related to the Job Shop problem under Resource Constraints (JSRC), where the jobs may require, besides the machines, some additional resources for their processing (see \citet{JWRJ86}). In such problems, conflicts arise when resource requirements exceed available capacities. These conflicts can occur at the operation level when different operations within a job have varying resource requirements. Alternatively, if all operations within a job have the same resource requirements, conflicts can be represented at the job level instead. Generally, conflicts in the JSRC problem are modeled using a conflict hypergraph. However, in some particular cases, a simple graph suffices, such as when only two machines are available or when each resource is limited to a single unit (see Section~\ref{sec:complexity}).
	
	\citet{blazewicz1983} expanded the three field classification, $ \alpha|\beta|\gamma $, to cover the resource-constrained scheduling problems. Parameter $\beta$ denotes the job and resource characteristics. The presence of additional resources is specified by $res\lambda\sigma\delta $, where $\lambda$, $\sigma$, $\delta$ $\in \{.,k\}$ represent respectively the number of resource types, resource availabilities and resource requirements. A dot ``.'' indicates that the corresponding parameter can take any integer value, whereas a positive integer $k$ indicates that the number of resource types is equal to $k$, each resource is available in the amount of $k$ units and the resource requirements of each operation are equal to at most $k$ units. We denote by $res^t\lambda\sigma\delta$ the special case where all operations within a job have the same requirements of each resource, and let \( R_s(J_j) \) be the resource requirement of job \( J_j \) for resource \( R_s \). Note that the scheduling problem $\alpha|res^t\lambda\sigma\delta|\gamma$ is a subproblem of $\alpha|res\lambda\sigma\delta|\gamma$.
	
In this paper, we investigate how the structure of the conflict
graph affects the computational complexity and approximability
of job shop scheduling. We first establish that the problem
is strongly NP-hard even on two machines with unit-time
operations and arbitrary conflict graphs. We then show that
the JSC problem is polynomially equivalent to
$J|res^t.11|C_{\max}$. For the two-machine case with unit-time
operations and a cluster conflict graph, we derive an explicit
formula for the optimal makespan and a linear-time algorithm.
This result implies that $J2|res^t111, p_{ij}=1|C_{\max}$
is polynomially solvable, in contrast to the known NP-hardness
of $J2|res111, p_{ij}=1|C_{\max}$.
Furthermore, we develop a polynomial-time $4/3$-approximation
algorithm for arbitrary conflict graphs when each job consists
of exactly two unit-time operations, together with improved approximation guarantees under
additional restrictions on the agreement graph.
This extends the corresponding approximation result for
the two-machine flow shop with conflict graphs to the job
shop setting, where both machine routes are allowed.
Finally, we complement our theoretical analysis with
mixed-integer linear programming formulations, lower bounds,
and a genetic algorithm for the general JSC problem.

The remainder of this paper is structured as follows. Section~\ref{sec:literature} reviews the related literature. Section~\ref{sec:complexity} discusses the complexity of the JSC problem. Section~\ref{sec:MILP} presents mathematical models, followed by lower bounds in Section~\ref{sec:LB}. Section~\ref{sec:GA} describes genetic algorithms. The computational experiments are detailed in Section~\ref{sec:CE}.
	\section{Literature review}\label{sec:literature}
	
	The basic job shop scheduling problem is among the most intractable machine scheduling problems, with only a few special cases that are solvable in polynomial time. The two-machine case can be solved in polynomial time when $m_j \leq 2$ ($J2|m_j \leq 2|C_{max}$)~\cite{jackson1956extension} by a modification of Johnson's algorithm for problem $F2||C_{max}$~\cite{johnson1954optimal}. Other cases that are tractable include $J2|r_j, p_{ij}=1|C_{max}$~\cite{kubiak1995efficient}, $J2|n=k | C_{max}$~\cite{brucker1994polynomial}, $J|n=2|C_{max}$~\cite{brucker1988efficient, akers1956graphical}. Minor modifications to these cases render the problem NP-hard. In particular, the problems $J2|m_j \leq 3|C_{max}$~\cite{lenstra1977complexity}, $J3|m_j \leq 2|C_{max}$~\cite{lenstra1977complexity}, $J2|p_{ij} \in \{1,2\}|C_{max}$~\cite{lenstra1979computational},  $J3|p_{ij}=1|C_{max}$~\cite{lenstra1979computational}, and  $J3|n = 3| C_{max}$~\cite{sotskov1995np} are NP-hard. When resource constraints are introduced, the simplest job shop variant, $J2|res111, p_{ij} = 1|C_{max}$, is already NP-hard in the strong sense~\cite{blazewicz1983}. 
	
	
	Scheduling with conflict graphs has been studied for parallel machines--specifically in the case of identical machines~\cite{mohabeddine2019new, tellache2025scheduling}--as well as for dedicated machines. In the latter case, most studies focus on open shop and flow shop problems with conflict graphs. To the best of our knowledge, the job shop problem with conflict graphs has not yet been addressed. In the following, we focus on the open shop and flow shop cases, as they are closely related to the job shop problem.
	
	The Open Shop with Conflict graph (OSC) problem with makespan minimization is strongly NP-hard for $m \geq 2$ \cite{TellacheBoudharDAM2017}. Various complexity results have been established for different problem settings. \citet{TellacheBoudharDAM2017} showed that the two-machine OSC problem is strongly NP-hard when $p_{ij} \in \{1,2,3\}$ and the conflict graph $G$ is the complement of a bipartite graph. However, when restricted to $p_{ij} \in \{0,1,2\}$, the problem becomes polynomially solvable for arbitrary conflict graphs. Moreover, preemption makes the two-machine problem tractable regardless of processing times and conflict structures. For unit-time operations, the problem remains polynomially solvable when increasing the number of machines to three for conflict graphs that are complements of triangle-free graphs, but becomes NP-hard for arbitrary graphs. For the resolution methods, the authors presented heuristic approaches and lower bounds for the general $m$-machine OSC problem. Later, \citet{TELLACHE2019154} addressed the proportionate two-machine OSC problem and proved that it is strongly NP-hard when processing times take two distinct values in  $ \{a, 2a+b\} $, ($ a \geq 1 $ and ($ b \geq 1 $ or $ -a<b<0 $)), for general conflict graphs. This result closes the complexity status of the two-machine OSC problem with two values of processing times. Further work by \citet{TELLACHE202185} demonstrated that the two-machine OSC problem remains NP-hard when the conflict graph is a complement of a tree, but is polynomially solvable for complements of caterpillars and cycles. Additionally, the author studied the OSC problem with release times, proving that the two-machine case is strongly NP-hard when release times belong to $\{0, r\}$ (for arbitrary $r$) and processing times are in $\{1,2\}$, for general conflict graphs. However, when the conflict graph is bipartite and $p_{ij} = 1$, the problem can be solved in polynomial time. More recently, \citet{tellache2023genetic} introduced a hybrid genetic algorithm with mathematical models and lower bounds. Their approach outperformed the two-phase heuristic approach of~\cite{TellacheBoudharDAM2017}, optimally solving 93.49\% of instances, with an average deviation of 0.475\% from the best-known lower bounds. 
	
	The Flow Shop with Conflict Graph (FSC) problem is more difficult than the OSC problem, with only a few cases solvable in polynomial time. \citet{tellache2018ANOR} proved that the two-machine FSC problem is strongly NP-hard, even when restricted to complements of complete split graphs and complements of complete bipartite graphs. This latter case remains true even if preemption is allowed. The authors also studied the case of unit-time operations and showed that it is strongly NP-hard, even for two machines, as it is polynomially equivalent to the Minimum Path Cover (MPC) problem. Consequently, the problem is NP-hard for graph classes where the MPC problem is NP-hard and solvable in polynomial time for graph classes where the MPC problem is tractable. In the special case of unit-time jobs (i.e., jobs with a single operation), the problem can be solved in $O(n^{2.5})$ time on two machines. The same study also introduced heuristic approaches and lower bounds for the general $m$-machine FSC problem. \citet{tellache2017two} further explored the two-machine FSC problem with unit-time operations, developing mathematical models and a branch-and-bound algorithm. For the same problem, \citet{cai2018approximation} proposed a $\frac{4}{3}$-approximation algorithm for arbitrary conflict graphs and a $\frac{3}{2}$-approximation algorithm for a conflict graph that is a complement of a complete bipartite graph. The two-machine FSC problem was also examined in \citet{tellache2016two}, where the authors demonstrated that permutation schedules are not dominant, even for two machines. However, they identified a special case where there exists an optimal schedule that is a permutation schedule. Finally, \citet{TELLACHE202185} showed that the two-machine FSC problem remains NP-hard even when restricted to complements of trees, even with preemption allowed.
	
	\section{Complexity analysis}
	\label{sec:complexity}
	
	The JSC problem is strongly NP-hard even on two machines and an edgeless
	conflict graph, since the classical two-machine job shop scheduling
	problem, $J2||C_{\max}$, is strongly NP-hard. A stronger result follows
	from the NP-hardness result for
	$F2|ConfG=(V,E),\,p_{ij}=1|C_{\max}$~\cite{tellache2018ANOR},
	which is a special case of JSC. Consequently, JSC remains strongly
	NP-hard even when all operations have unit processing times and all
	jobs consist of exactly two operations following a common machine route.

	Since conflicts between jobs may arise from shared resources, we first
	establish the relationship between the JSC and JSRC problems in the
	following theorem. Although the proof shares similarities with those
	in~\cite{TellacheBoudharDAM2017} and~\cite{tellache2018ANOR}, we
	include it here for completeness. We then analyse the complexity of
	JSC under several special cases.
	
	\begin{theorem} \label{thm:JSCJSRC}
		The JSC problem is polynomially equivalent to $J|res^t.11|C_{max}$.
	\end{theorem}
	\begin{proof}
		Given any instance of JSC, we construct an instance of $J|res^t.11|C_{max}$ as follows. The set of jobs is the same for the two problems. For every edge $e_s = \{J_k, J_l\} \in E$, we introduce a resource $R_s$ with availability $|R_s| = 1$. The requirements of $R_s$ are such that $R_s(J_k) = R_s (J_l) = 1$ and $R_s(J_j) = 0$, for all $j\notin \{k, l\}$.
		
		Conversely, for a given instance of $J|res^t.11|C_{max}$, we construct an instance of JSC as follows. The set of jobs is the same for the two problems. The conflict graph $G$ is constructed based on the relation: $\{J_k,J_l\} \in E$ if and only if there exists at least one resource $R_s$ such that $R_s (J_k)+R_s (J_l) > 1$.
		
		Thus, two jobs can be processed simultaneously in the JSC problem if and only if they do not share the same resource in the problem $J|res^t.11|C_{max}$ and the theorem follows.
	\end{proof}

	The conflict graph associated with the case
	$J|res^t.11|C_{\max}$ can be characterized as follows. For each
	resource $R_r$, let $C_r$ denote the set of jobs requiring that
	resource. The edge set of the conflict graph satisfies
	\begin{equation}\label{eq:resource-graph}
		E(G)=\bigcup_{r=1}^{s}\binom{C_r}{2},
	\end{equation}
	where $\binom{C_r}{2}$ denotes the set of all unordered pairs of
	distinct jobs in $C_r$. Since the sets $C_r$ may overlap, the
	corresponding cliques may share vertices. Thus, the cliques induced
	by the resource sets form an edge clique cover of the conflict graph.
	
	The simplest case of the JSRC problem,
	$J2|res111,\,p_{ij}=1|C_{\max}$, is already known to be NP-hard
	~\cite{blazewicz1983}. However, the reduction used to establish this
	result relies on instances in which job operations have different
	resource requirements. Consequently, this result does not directly
	establish the NP-hardness of
	$J2|res^t111,\,p_{ij}=1|C_{\max}$. We show in the following that $J2|res^t111,\,p_{ij}=1|C_{\max}$ can be solved in polynomial time.
	
	On the other hand, Szwarcfiter~\cite{szwarcfiter1987} provides a
	polynomial-time algorithm for unit-time operations when the numbers
	of machines, resources, and operations per job, as well as the
	number of distinct resource demand values, are fixed. This result
	therefore yields a polynomial-time solvable case of our problem
	when the operations have unit processing times, the numbers of
	machines and operations per job are bounded, and the conflict graph
	admits an edge clique cover of bounded size, provided that the
	number of resources is also bounded.
	
	We next show that the two-machine JSC problem with unit-time
	operations is polynomial-time solvable when the conflict graph is
	a cluster graph, i.e., a disjoint union of cliques. Neither the
	number of cliques nor the number of operations per job is fixed.
	This includes complements of complete split graphs and yields a
	tractable JSRC case with arbitrarily many resources, as explained
	at the end of the following subsection. Finally, we present a polynomial-time $4/3$-approximation algorithm for the NP-hard two-machine JSC problem
	with exactly two unit-time operations per job. The guarantee improves to $1+1/g$
	when every cycle of the agreement graph whose jobs all
	have the same route has length at least $g\ge3$. Moreover, these approximation guarantees are tight for the algorithm as stated. 
	
\subsection{Polynomial algorithm for cluster conflict graphs}

Consider a two-machine JSC instance with unit-time operations whose conflict graph is a cluster graph with cliques \(Q_1,\ldots,Q_q\), where isolated vertices are treated as singleton cliques. Let \(L_k\) denote the load of machine \(M_k\), and let $L_{Q_r}=\sum_{J_j\in Q_r} m_j$
denote the load of clique \(Q_r\). For a nonempty instance, define \(\delta=1\) if \(L_1=L_2\) and all first operations require the same machine, and \(\delta=0\) otherwise.

\begin{theorem}\label{thm:cluster}
The two-machine JSC problem with unit-time operations and a cluster conflict graph can be solved in $O(N)$ time, given the clique partition, where $N=\sum_{j=1}^{n}m_j$. The optimal makespan is given by
	\begin{equation}\label{eq:cluster-optimum}
		C_{\max}^{*}
		=
		\max\left\{
		\max_{1\le r\le q}L_{Q_r},
		\ \max\{L_1,L_2\}+\delta
		\right\}.
	\end{equation}
\end{theorem}

\begin{proof}[Proof of Theorem~\ref{thm:cluster}]
We prove this theorem by construction. We first show that the right-hand side of~\eqref{eq:cluster-optimum} is a lower bound on the makespan.

\begin{lemma}\label{lem:cluster-lower-bound}
	Every feasible schedule satisfies
	\[
	C_{\max}\ge
	\max\left\{
	\max_{1\le r\le q}L_{Q_r},
	\ \max\{L_1,L_2\}+\delta
	\right\}.
	\]
\end{lemma}

\begin{proof}
	The operations assigned to the same machine cannot overlap,
	nor can the operations belonging to jobs in the same clique.
	Thus, each machine load and each clique load is a lower bound
	on the makespan. If $\delta=1$, all jobs start on the same
	machine. The other machine cannot process any operation
	before time~1, since its operations have predecessors of
	unit processing time. It must then process $L_1=L_2$
	operations, giving $C_{\max}\ge\max\{L_1,L_2\}+1$.
	Combining these bounds proves the claim.
\end{proof}

Let $B(I)$ denote the right-hand side of~\eqref{eq:cluster-optimum} for an instance $I$, with $B(\varnothing)=0$. By Lemma~\ref{lem:cluster-lower-bound}, $C_{\max}^{*}\ge B(I)$. In the following, we construct a feasible schedule with makespan $B(I)$, which is therefore optimal. The construction proceeds one unit-time slot at a time.

Let $R$ denote the remaining instance after scheduling a time slot in which completed jobs are removed, and unfinished jobs retain their remaining operation sequences. Initially $R=I$. If $R$ is empty, the construction terminates. If $B(R)=1$, then at most two unfinished jobs remain, each with one operation. If two jobs remain, their operations require different machines and belong to different cliques. They can therefore be processed simultaneously in one
time slot.

Otherwise, let $H=B(R)\ge2$. We call the first remaining operation of each unfinished job \emph{available}, and a machine or clique \emph{tight} if its remaining load equals $H$.
A \emph{feasible selection} consists of one available operation, or two available operations on different machines and in different cliques. The idea of the proof is to first choose a feasible selection $S$ serving every tight machine and clique, and then, if necessary, repair it before appending the slot so that the remaining bound is at most $H-1$.

\begin{lemma}\label{lem:cluster-selection}
Let $R$ be a remaining instance with $H=B(R)\ge2$.
There exists a feasible selection serving every tight
machine and clique.
\end{lemma}

\begin{proof}
Every tight machine has an available operation. Otherwise, all unfinished jobs have their next operation on the other machine. By alternation, each remaining job route contains at
least as many operations on that machine as on the tight machine, so the other machine also has remaining load at least $H$. Both machine loads would then equal $H$, with $\delta=1$ for $R$, contradicting $B(R)=H$.

If both machines are tight, each has at least one available operation. Choose a tight clique $Q_r$ if one exists, and any clique with unfinished jobs otherwise. This clique contains an available operation. Since its remaining load is at most $H$, whereas the total remaining load is $2H$, there is also an unfinished job outside $Q_r$, whose first remaining operation is available.

Choose an available operation in $Q_r$, say on $M_1$. If an available operation outside $Q_r$ requires $M_2$, these two operations form a feasible selection. Otherwise, all available operations outside $Q_r$ require $M_1$. Since $M_2$ has an available operation, that operation must belong to $Q_r$; pair it with any available operation outside $Q_r$.

In either case, the selected operations require different machines and belong to different cliques. Thus, the selection serves both machines and $Q_r$. It also serves every tight clique: if two cliques are tight, their combined load is $2H$, so they contain all remaining operations,
and the operation selected outside $Q_r$ belongs to the other tight clique.

If at most one machine is tight, the total remaining load is less than $2H$, so at most one clique is tight.

Suppose that both a tight machine, say $M_1$, and a tight clique $Q_r$ exist. If $Q_r$ contains an available operation requiring $M_1$, select that operation; it serves both the tight machine and the tight clique. Otherwise, all available operations in $Q_r$ require $M_2$. Since $M_1$
has an available operation, that operation must belong to a clique other than $Q_r$. Pair it with any available operation in $Q_r$. The two operations require different machines and belong to different cliques, so this selection is feasible and serves both tight constraints.

If a tight machine exists but no clique is tight, select any available operation on that machine. If a tight clique exists but neither machine is tight, select any available operation in that clique. Finally, if neither a machine nor a clique is tight, any available operation suffices.
\end{proof}

\begin{lemma}\label{lem:cluster-repair}
Let $R$ be a remaining instance with $H=B(R)\ge2$, and let
$S$ be a feasible selection serving every tight machine
and clique. By replacing one selected operation by another
available operation, or adding one available operation, if
necessary, $S$ can be transformed into a feasible selection
whose processing leaves a remaining bound at most $H-1$.
\end{lemma}

\begin{proof}
Let $R'$ be the instance that would remain after processing
$S$. Since $S$ serves every tight machine and clique, all
machine and clique loads in $R'$ are at most $H-1$.
If $B(R')\le H-1$, keep $S$ unchanged.

Suppose instead that $B(R')>H-1$. Then both remaining
machine loads equal $H-1$, and all available operations in
$R'$ require the same machine, say $M_b$. Equality of the machine loads implies that every remaining route has even length. Thus, every unfinished job in $R'$ has at
least two remaining operations.

Let $J_u$ and $J_v$ denote the jobs whose operations in $S$
require $M_b$ and the other machine, respectively, if present.
If $J_u$ exists, it has exactly one operation remaining in
$R$; otherwise, processing its selected operation would make
its successor available on the other machine in $R'$.
We now choose an unselected job $J_w$ to
provide an operation on $M_b$ for the repaired selection.

If $J_u$ exists and its clique is tight, choose a job $J_w$
that remains unfinished in that clique. Such a job exists
because the clique retains $H-1>0$ operations in $R'$.
Otherwise, choose an unfinished job $J_w$ outside the clique
of $J_v$ when $J_v$ exists, and any unfinished job in $R'$
when $J_v$ is absent. A job outside the clique of $J_v$
exists whenever needed, since $R'$ contains
$2(H-1)$ operations, while that clique contains at most
$H-1$.

The job $J_w$ was not selected in $S$: $J_u$, if present,
has finished in $R'$, whereas $J_v$, if present, belongs
to a different clique. Thus, $J_w$ has the same remaining
route in $R$ and $R'$. Its first operation is therefore
already available on $M_b$ in $R$, and it has at least
two remaining operations.

If $J_u$ exists, replace its operation by the available
operation of $J_w$; otherwise, add that operation to $S$.
The selection is feasible because the operations of $J_w$
and $J_v$, when both are selected, require different
machines and belong to different cliques. Every tight
machine and clique remains served: any removed operation
is replaced on the same machine and, if its clique is
tight, within the same clique. Thus, all remaining loads
are still at most $H-1$.

If the operation of $J_u$ was replaced, it remains
available on $M_b$, while processing $J_w$ makes its
successor available on the other machine. Both machines
therefore have available operations. If $J_u$ was absent,
adding the operation of $J_w$ leaves machine loads $H-2$
and $H-1$. In either case, $\delta=0$ for the corrected
remaining instance, so its bound is at most $H-1$.
\end{proof}

For each remaining instance with $B(R)\ge2$, choose a
selection using Lemma~\ref{lem:cluster-selection} and repair
it, if necessary, using Lemma~\ref{lem:cluster-repair}.
Append the resulting slot and repeat with the remaining
instance, handling $B(R)=1$ as described above.

The remaining instance still has
a cluster conflict graph, so the construction can be
repeated. It yields a feasible schedule in at most $B(I)$
slots, since each appended slot decreases $B(R)$ by at
least one. Lemma~\ref{lem:cluster-lower-bound} then proves
optimality and establishes~\eqref{eq:cluster-optimum}.

To implement the construction, divide unfinished jobs into
four categories according to their next machine and whether
they have one or at least two operations remaining. In each
nonempty category, retain one representative job from each
of two distinct cliques, or from its sole clique if only one
is represented. Additionally, retain one representative from
each nonempty category of every tight clique. Since at most
two cliques are tight, at most sixteen jobs are retained.

These representatives suffice for both lemmas. A selected
job in a tight clique can be replaced by that clique's
representative in the same category. If exactly one selected
clique is tight, the other category has a representative
outside it: either it spans two cliques, or its sole clique
is already compatible with the original selection.
If neither selected clique is tight, compatible
representatives exist unless both categories contain jobs
only in the same single clique, which would also preclude
the original selection. Singleton selections are immediate.
The repair similarly requires a job in a given category
either inside a tight clique or outside one specified clique.
Thus, only a constant number of candidates must be examined
at each step.

Maintain job lists for each clique--category pair and, for
each category, a list of the cliques represented in it.
Also maintain the machine loads, the numbers of available
operations on each machine, and buckets of cliques indexed
by remaining load. Each slot changes at most two jobs and
two clique loads, requiring constant-time list updates.
The largest nonempty load bucket is maintained by a pointer
that moves only downwards, through at most $N$ positions
over the whole construction. Initialization takes $O(N)$
time, and at most $N$ slots are constructed. Hence, the
total running time is $O(N)$.
\end{proof}

If the conflict graph is supplied by adjacency lists,
identifying its clique components adds $O(n+|E|)$ time.
	
\begin{corollary}\label{cor:cluster-jsrc}
	$J2|res^{t}.11,\,p_{ij}=1|C_{\max}$
	is polynomially solvable when each job requires at most
	one resource. In particular,
	$J2|res^{t}111,\,p_{ij}=1|C_{\max}$
	is polynomially solvable.
\end{corollary}

\begin{proof}
Consequence of Theorems~\ref{thm:JSCJSRC} and~\ref{thm:cluster}.
\end{proof}
\subsection{Approximation algorithms}

We consider the two-machine JSC problem in which each job consists
of two unit-time operations, one on each machine. For each job $J_j$,
let $A_j$ and $B_j$ denote its operations on $M_1$ and $M_2$,
respectively.

A $4/3$-approximation algorithm is known for the two-machine FSC
problem with unit-time operations~\cite{cai2018approximation,chen2022pathcover}.
We show that the same approximation ratio can be achieved for
the JSC problem when both job routes are allowed.

\begin{theorem}\label{thm:four-thirds}
The two-machine JSC problem with
exactly two unit-time operations per job admits a polynomial-time $4/3$-approximation algorithm.
\end{theorem}
\begin{proof}
Let $\mathcal H$ be the multigraph obtained from the agreement
graph $\overline G$ by retaining one copy of each edge joining
jobs with the same route and two parallel copies of each edge
joining jobs with different routes. Find a maximum-cardinality subset of edges such that every vertex has degree at most two.
Let $F$ be the spanning submultigraph of $\mathcal H$ induced
by the selected edges, and set $r=|E(F)|$. This is a maximum-cardinality $b$-matching problem with
$b(v)=2$ for all vertices, which can be solved in polynomial
time~\cite[Chapter~32]{schrijver2003}.
	
We first establish a lower bound on the optimal makespan.
Consider an optimal schedule, and let $a$ denote the number
of slots containing two operations. Joining the two jobs
processed together in each such slot defines a multigraph
on the jobs. Its maximum degree is at most two, since each
job has exactly two operations. Two jobs with the same route cannot be paired twice: their precedence constraints would require each paired slot to precede the other. Moreover, every pair consists of nonconflicting jobs. Thus, the resulting multigraph is a subgraph of $\mathcal H$ with maximum degree at most two, and the definition of $r$ gives $a\le r$. The $a$ paired slots contain $2a$ operations, while the
remaining $2n-2a$ operations occupy singleton slots.
Counting these $2n-a$ nonempty slots therefore yields
\begin{equation}\label{eq:bmatching-lb}
	C_{\max}^{*}\ge 2n-a\ge 2n-r.
\end{equation}
	
We now construct a schedule from $F$. Its components are paths and cycles, including isolated vertices and cycles of two parallel edges. Delete one edge from every cycle whose jobs all have the same route, and let $k$ be the number of deleted edges. Orient each path from one endpoint to the other and each cycle in a cyclic direction, so that every vertex has indegree and outdegree at most one. For every arc $i\to j$, put $A_i$ and $B_j$ together in a block, and place each unpaired operation in a singleton block. Each operation belongs to exactly one block because every vertex has indegree and outdegree at most one. Each block can be processed feasibly in one time slot.
	
To order these blocks, add a precedence arc from the block containing each job's first operation to the block containing its second operation. For a path component, the underlying block graph is a path and therefore has no directed cycle. For a cycle of length at least three, opposite job routes give opposite directions to the corresponding precedence arcs around the block cycle. A directed cycle therefore occurs only when all jobs have the same route; such cycles have been broken. Finally, a retained two-edge cycle joins jobs with opposite routes: one block contains both first operations and the other both second operations. Hence, the block graph is acyclic.
	
Process the blocks consecutively in a topological order.
This gives a feasible schedule. The $r-k$ retained edges
yield $r-k$ two-operation blocks, while the remaining
$2n-2(r-k)$ operations form singleton blocks. Hence,
\[
C_{\max}^{\mathrm{alg}}
=(r-k)+2n-2(r-k)=2n-r+k.
\]

Every deleted cycle has at least three edges, and these cycles are disjoint, so $3k\le r$. Moreover, $r\le n$	by the degree constraint, while $C_{\max}^{*}\ge n$ by the machine loads. Together with~\eqref{eq:bmatching-lb}, these inequalities give
	\[
	C_{\max}^{\mathrm{alg}}
	=2n-r+k
	\le C_{\max}^{*}+\frac n3
	\le \frac43 C_{\max}^{*}.
	\]
	
Constructing $\mathcal H$ takes $O(n^2)$ time. After the polynomial-time matching computation, the remaining steps take $O(n)$ time. Thus, the algorithm is polynomial.
\end{proof}

The same computation provides the certificate
\begin{equation}\label{eq:rounding-certificate}
 2n-r\le C_{\max}^{*}\le 2n-r+k.
\end{equation}
In particular, the constructed schedule is optimal whenever
the selected maximum subgraph contains no cycle whose jobs
all have the same route.

\begin{corollary}\label{cor:girth-approx}
If every cycle of the agreement graph whose jobs all have the same route has length at least $g$, for some integer $g\ge3$, the same algorithm is a $(1+1/g)$-approximation. In particular, the guarantee is $5/4$ when neither route class induces a triangle in the agreement graph.
\end{corollary}
\begin{proof}
Every deleted cycle has at least $g$ edges, so $k\le r/g$.
Using $r\le n$ and~\eqref{eq:bmatching-lb}, we obtain
\[
 C_{\max}^{\mathrm{alg}}
 \le 2n-r+\frac r g
 \le \left(1+\frac1g\right)(2n-r)
 \le \left(1+\frac1g\right)C_{\max}^{*}.
\]
Take $g=4$ for the triangle-free case. If both route-induced
agreement graphs are forests, no edge is deleted and the
algorithm is exact.
\end{proof}

The factor $1+1/g$ is tight for the algorithm as stated. Indeed, consider an agreement graph consisting of two disjoint $g$-cycles, one in each job route class, joined by a perfect matching. In $\mathcal H$, selecting either the two cycles or both copies of every matching edge gives a maximum subgraph: both choices have $2g$ edges, attaining the bound $r\le n=2g$. The first choice requires deleting one edge from each cycle and produces a schedule of length $2g+2$. The second requires no deletion and produces a schedule of length $2g$, which is optimal by the machine-load lower bound. Thus, an allowed choice of $F$ attains the ratio $(2g+2)/(2g)=1+1/g$.

	\section{Mathematical formulations}\label{sec:MILP}
	In this section, we present four different MILPs for $ Jm|ConfG = (V,E)|C_{max} $, adapted from existing models in the literature to take into account the conflict constraints between jobs.
	
	\subsection{Precedence relationship-based formulation}
	\label{first_preced}
	This model builds on the classical formulation of \citet{manne1960}, which accounts for the precedence relationships between operations, and introduces the following decision variables. The variable $S_{ij}$ denotes the start time of operation $O_{ij}$. The binary variable $y_{O_{ij}O_{uk}}$ equals 1 if operation $O_{ij}$ is scheduled before operation $O_{uk}$ on the same machine (i.e., $\mu_{ij} = \mu_{uk}$), and 0 otherwise; note that we may have $J_j = J_k$ since a job can visit the same machine multiple times. The binary variable $z_{O_{ij}O_{uk}}$ equals 1 if operation $O_{ij}$ is scheduled before operation $O_{uk}$ when jobs $J_j$ and $J_k$ are in conflict, and 0 otherwise. The formulation is summarized in (\mFA).  
	
	\begin{subequations}\label{P1}
		\begin{align}
			(\mFA)\qquad \min \quad & C_{max} \label{F1:obj}\\ 
			s.t. \quad & S_{i+1,j} \geq S_{ij} + p_{ij}, \quad && \forall J_j \in J, i \in \{1,\ldots,m_j-1\} \label{F1:precedence1}\\
			& C_{max} \geq S_{m_j,j} + p_{m_j,j}, \quad &&\forall J_j \in J \label{F1:makespan1}\\
			& S_{uk} \geq S_{ij} + p_{ij} - M(1-y_{O_{ij}O_{uk}}), \quad&& \forall O_{ij}\neq O_{uk}: \mu_{ij} = \mu_{uk}, j<k \label{F1:machine1a}\\
			& S_{ij} \geq S_{uk} + p_{uk} - My_{O_{ij}O_{uk}},\quad&& \forall O_{ij}\neq O_{uk}: \mu_{ij} = \mu_{uk}, j<k \label{F1:machine1b}\\
			& S_{uk} \geq S_{ij} + p_{ij} - M(1-z_{O_{ij}O_{uk}}), \quad &&  \forall O_{ij}, O_{uk}: \{J_j,J_k\} \in E, j<k, \mu_{ij} \neq \mu_{uk} \label{F1:conflict1a} \\
			& S_{ij} \geq S_{uk} + p_{uk} - Mz_{O_{ij}O_{uk}}, \quad &&  \forall O_{ij}, O_{uk}: \{J_j,J_k\} \in E, j<k, \mu_{ij} \neq \mu_{uk} \label{F1:conflict1b}\\
			&y_{O_{ij}O_{uk}} \in \{0,1\}, && \forall O_{ij}\neq O_{uk}: \mu_{ij} = \mu_{uk}, j<k\label{F1:vary} \\
			&z_{O_{ij}O_{uk}} \in \{0,1\}, &&\forall O_{ij}, O_{uk}: \{J_j,J_k\} \in E, j<k, \mu_{ij} \neq \mu_{uk} \label{F1:varz} \\
			&C_{max}, S_{ij} \geq 0, && \forall O_{ij}.\label{F1:vars} 
		\end{align}
	\end{subequations}
	
	The objective function~\eqref{F1:obj} minimizes the makespan. Constraints~\eqref{F1:precedence1} enforce the precedence relationships between the operations of each job. Constraints~\eqref{F1:makespan1} equate the makespan to the maximum completion times of the last operation of each job. Constraints~\eqref{F1:machine1a} and~\eqref{F1:machine1b} are paired disjunctive constraints on two operations of the same machine, ensuring that $O_{ij}$ either precedes $O_{uk}$ or (exclusive) follows it. Finally, Constraints~\eqref{F1:conflict1a} and~\eqref{F1:conflict1b} are paired disjunctive constraints on operations of conflicting jobs, ensuring that such operations cannot be processed simultaneously on different machines.
	
	An alternative formulation can be obtained by exploiting the symmetry between the two possible orderings of each pair of operations. More specifically, for each pair of disjunctive constraints in $(\mFA)$, we retain the first inequality and replace the second one by an equality involving two symmetric binary variables. For the pair~\eqref{F1:machine1a}--\eqref{F1:machine1b}, constraint~\eqref{F1:machine1a} ensures that if $y_{O_{ij}O_{uk}}=1$, then $O_{ij}$ precedes $O_{uk}$ on their common machine. If $y_{O_{ij}O_{uk}}=0$, then $O_{uk}$ must precede $O_{ij}$, which is ensured by introducing the symmetric variable $y_{O_{uk}O_{ij}}$ and imposing
	
	\begin{subequations}
		\begin{align} 
			y_{O_{ij}O_{uk}}+y_{O_{uk}O_{ij}} &= 1,
			&& \forall O_{ij}\neq O_{uk}: \mu_{ij}=\mu_{uk},\ j\neq k
		\end{align} 
	\end{subequations}
	
	The pair~\eqref{F1:conflict1a}--\eqref{F1:conflict1b} is treated similarly by introducing the symmetric variable $z_{O_{uk}O_{ij}}$ and imposing
	
	\begin{subequations}
		\begin{align} 
			z_{O_{ij}O_{uk}}+z_{O_{uk}O_{ij}} &= 1,
			&& \forall O_{ij},O_{uk}: \{J_j,J_k\}\in E,\ j\neq k,\ \mu_{ij}\neq\mu_{uk}
		\end{align} 
	\end{subequations}
	
	Accordingly, formulation $(\mFB)$ is obtained from $(\mFA)$ by replacing the pairs~\eqref{F1:machine1a}--\eqref{F1:machine1b} and~\eqref{F1:conflict1a}--\eqref{F1:conflict1b} by the corresponding big-$M$ inequalities and symmetry constraints described above, and by defining the additional symmetric binary variables. The remaining constraints of $(\mFA)$ are unchanged.

	\subsection{Time-indexed formulation}
	The idea of time-indexed variables was first proposed by \citet{bowman1959schedule} and subsequently applied to a variety of scheduling problems. We adapt this approach to our setting. Our formulation uses a binary variable $x_{ijt}$, which equals 1 if operation $O_{ij}$ starts at time $t$, and 0 otherwise. Let $T$ be an upper bound on the job completion times. In the computational experiments, we set $T$ to the makespan of the schedule obtained from the warm-start procedure described in Section~\ref{sec:warm_start}. The resulting formulation is presented in~(\mFC).
	
	\begin{subequations}\label{P2}
		\small
		\begin{align}
			(\mFC)\qquad \min \quad & C_{max} \label{F2:obj}\\ 
			s.t. \quad & \sum_{t=0}^{T-p_{ij}} x_{ijt} = 1, \quad && \forall O_{ij}\label{F2:assign4}\\
			& \sum_{t=0}^{T-p_{i+1,j}} t x_{i+1,jt} \geq \sum_{t=0}^{T-p_{ij}} t  x_{ijt}+ p_{ij}, \quad && \forall J_j \in J, i \in \{1,\ldots,m_j-1\} \label{F2:precedence4}\\
			& \sum_{O_{ij}:\mu_{ij}=u} \sum_{\tau=\max\{0,t-p_{ij}+1\}}^{\min \{t,T-p_{ij} \}} x_{ij\tau} \leq 1, \quad && \forall u \in M,t \in \{0,\ldots,T-1\} \label{F2:machine4}\\
			& \sum_{i=1}^{m_j} \sum_{\tau=\max\{0,t-p_{ij}+1\}}^{\min\{t,T-p_{ij}\}} x_{ij\tau} + \sum_{u=1}^{m_k} \sum_{\tau=\max\{0,t-p_{uk}+1\}}^{\min\{t,T-p_{uk}\}} x_{uk\tau} \leq 1,  && \forall \{J_j,J_k\}\in E,\; j<k, \notag \\ 
			&& &  t\in\{0,\ldots,T-1\} \label{F2:conflict4} \\
			& C_{max} \geq \sum_{t=0}^{T-p_{m_j,j}} (t + p_{m_j, j}) x_{m_j, j t}, \quad && \forall J_j \in J\label{F2:makespan4}\\
			& x_{ijt} \in \{0,1\}, \quad &&\forall O_{ij},t \in \{0,\ldots,T-p_{ij}\}. \label{eq:binary4}
		\end{align}
	\end{subequations}
	
	The objective function~\eqref{F2:obj} minimizes the makespan. Constraints~\eqref{F2:assign4} ensure that each operation starts exactly once. Constraints~\eqref{F2:precedence4} enforce the precedence relationships between operations of the same job. Constraints~\eqref{F2:machine4} ensure that each machine processes at most one operation at a time. Constraints~\eqref{F2:conflict4} prevent operations of conflicting jobs from being scheduled simultaneously. Finally, Constraints~\eqref{F2:makespan4} equate the makespan to the maximum completion time of the last operation of each job. Note that for $x_{ijt}$, we can set $x_{ijt} = 0$ for $t < \sum_{u=1}^{i-1} p_{uj}$ and $t > T - \sum_{u=i}^{m_j} p_{uj}$, due to the precedence constraints among the operations of each job.
	
	Constraints~\eqref{F2:conflict4}, which prevent conflicting jobs from being processed simultaneously, can be reformulated by exploiting the conflict degree of each job. For a given job $J_j$, the conflict degree corresponds to the degree of its associated vertex in the conflict graph. We denote it by $\mathrm{deg}_G(J_j)$, and let $N_G(J_j)$ represents the set of neighbors of $J_j$ in $G$. The idea is that, at any time $t$, if a job $J_j$ is being executed, then none of its neighboring jobs in $G$ can be processed at that time. Accordingly, Constraints~\eqref{F2:conflict4} can be equivalently expressed as:
	\begin{subequations}\label{P3}
		\begin{align}   
			\mathrm{deg}_G(J_j)
			\sum_{i=1}^{m_j} \sum_{\tau=\max\{0,\,t-p_{ij}+1\}}^{\min\{t,\,T-p_{ij}\}} x_{ij\tau}
			+ \sum_{k \in N_G(J_j)} \sum_{u=1}^{m_k} 
			\sum_{\tau=\max\{0,\,t-p_{uk}+1\}}^{\min\{t,\,T-p_{uk}\}} x_{uk\tau}
			\leq \mathrm{deg}_G(J_j),
			\quad \forall J_j \in J, \nonumber\\ t \in \{0,\ldots,T-1\}. \label{F3:conflict}
		\end{align}
	\end{subequations}
	
	This formulation requires $O(|V|T)$ constraints, whereas Constraints~\eqref{F2:conflict4} involve $O(|E|T)$ constraints. 
	We denote the model obtained by replacing Constraints~\eqref{F2:conflict4} in \mFC{} with Constraints~\eqref{F3:conflict} as \mFD{}.
	
	\subsection{Warm start} \label{sec:warm_start}
	We observed from the computational experiments that for large instances, the mathematical formulations may fail to find a feasible solution within the time limit. We introduce in this section a fast heuristic that provides a warm start for the solver, thereby helping it produce an upper bound within the time limit.
	
	\begin{algorithm}[H]
		\caption{Heuristic for JSC warm-starting}\label{heuristic_warmstart} 
		\hspace*{\algorithmicindent} \textbf{Inputs}: processing times $(p_{ij})_{1\leq i \leq m_j, 1 \leq j \leq n}$, adjacency matrix $A$ of $G$ 
		\begin{algorithmic}[1]
			\State Initialize $\pi'$ with the first operations of each job;
			\State Initialize $\mathcal U $ to all the operations of all the jobs;
			\State Initialize the earliest starting times $(s_{ij})_{1\leq i \leq m_j, 1 \leq j \leq n}$ of the operations of $\mathcal U$ to zero;
			\While {$\pi' \neq \emptyset $} 
			\State Select an operation $O_{ij}$ of $\pi'$ following a priority rule;\label{WS:priorityrule}
			\State $c_{ij}=s_{ij}+p_{ij}$;
			\State $\pi'=\pi'\setminus \{O_{ij}\}$,  $\mathcal U=\mathcal U\setminus \{O_{ij}\}$;
			\State Insert the next operation of the job corresponding to $O_{ij}$ (if any) into $\pi'$;
			\For{ $O_{i'j'} \in \mathcal U$  such that
				$O_{i'j'}\sim O_{ij}$}\label{WS:Fornondelay}
			\If{$s_{i'j'} <  s_{ij}+p_{ij}$ }
			\State $s_{i'j'}=s_{ij}+p_{ij}$;\label{WS:endFornondelay}
			\EndIf
			\EndFor
			\EndWhile
		\end{algorithmic}
		\hspace*{\algorithmicindent} \textbf{Outputs}: completion times $(c_{ij})_{1\leq i \leq m_J, 1 \leq j \leq n}$ and the makespan.
	\end{algorithm}
	
	The heuristic introduced in this section is detailed in Algorithm~\ref{heuristic_warmstart} and extends classical list scheduling approaches~\citep{panwalkar1977survey} to account for conflict constraints. At each iteration, it selects in Step~\ref{WS:priorityrule} one operation from a candidate set that includes, for every job, its next unscheduled operation. The selection follows a priority rule that chooses the operation with the earliest starting time; ties are broken using a priority metric reflecting the remaining workload of the job and its conflicting jobs. For an operation $O_{ij}$, this metric is computed as follows, where $\mathcal{S}_j$ denotes the set of unscheduled operations of job $J_j$:
	
	$$\beta_{ij}= \underbrace{\sum_{\substack{O_{uj} \in \mathcal{S}_j}} p_{uj}}_{\text{Current job workload}} + \underbrace{\sum_{j' \in N_G(j)} \sum_{\substack{O_{uj'} \in \mathcal{S}_{j'}}} p_{uj'}}_{\text{Conflicting jobs workload}}$$
	
	This metric prioritizes operations associated with jobs that either have a large remaining workload or are highly constrained by conflicts. Scheduling these operations earlier enhances the likelihood of executing them in parallel with other operations. In contrast, delaying them limits this possibility, potentially increasing the makespan.
	
	To satisfy the conflict constraints, once an operation is selected, all operations that conflict with it--either belonging to the same job or to conflicting jobs--whose start times are earlier than its completion time are postponed. Their start times are updated to match the completion time of the selected operation, thereby ensuring that no conflicting operations are executed simultaneously (Steps~\ref{WS:Fornondelay}-~\ref{WS:endFornondelay} of Algorithm~\ref{heuristic_warmstart}).
	
	To enhance the warm-start heuristic, we consider eight additional job-priority rules to select the operation to be scheduled at Step~\ref{WS:priorityrule} of Algorithm~\ref{heuristic_warmstart}. For each job $J_j$, let $P_j=\sum_{i=1}^{m_j}p_{ij}$ denote its total processing time, $p_j^{\max}$ and $p_j^{\min}$ its maximum and minimum operation processing times, $\bar p_j=P_j/m_j$ its mean operation processing time, and $m_j$ its number of operations. The SPT and LPT rules order jobs by nondecreasing and nonincreasing $P_j$, respectively. MWR and LWR order jobs by nonincreasing $p_j^{\max}$ and nondecreasing $p_j^{\min}$, respectively, while MOR and LOR order jobs by nonincreasing and nondecreasing $m_j$. Finally, MPT and descending MPT order jobs by nondecreasing and nonincreasing $\bar p_j$, respectively. Ties are broken by job index, increasingly for nondecreasing rules and decreasingly for nonincreasing rules.
	
	Each rule induces a priority ordering of the jobs, which is used to select operations from the candidate list in Step~\ref{WS:priorityrule}. The resulting job orderings, including the one derived from $\beta_{ij}$, are subsequently used in Section~\ref{meta_initialization} to construct seed chromosomes for the genetic algorithm.
	
	\section{Lower bounds}\label{sec:LB}
	We present in this section three classes of lower bounds on the makespan for the general $m$-machine JSC problem $Jm\,|\,ConfG=(V,E)\,|\,C_{\max}$. These bounds are incorporated into the stopping criteria of the genetic algorithms and are used to benchmark the solution quality of both the genetic algorithms and the mathematical formulations. They are adapted from~\cite{tellache2023genetic}, which proposed bounds for the OSC problem, and have been extended here to the JSC problem.
	
	\subsection{Conflict constraint relaxation-based lower bounds}
	Relaxing the conflicts between jobs reduces the problem to the basic job shop scheduling problem. Therefore, the lower bounds applicable to the basic job shop scheduling problem also extend to $Jm|ConfG=(V,E)|C_{max}$. One such bound is the classical one, which is derived from the maximum machine load and the total processing time of the jobs.
	
	\begin{equation*}
		LB_1 =
		\max\left\{
		\max_{j=1,\ldots,n}
		\sum_{i=1}^{m_j} p_{ij},
		\quad
		\max_{k=1,\ldots,m}
		\sum_{j=1}^{n}
		\sum_{\substack{i=1,\ldots,m_j\\ \mu_{ij}=k}} p_{ij}
		\right\}.
	\end{equation*}
	
	\subsection{Weighted agreement graph-based lower bound}\label{LBIS}
	The idea behind this lower bound is to identify a set of conflicting jobs or operations with the maximum total processing time. This total serves as a lower bound on the makespan, as conflicting jobs or operations must be scheduled in disjoint time intervals. 
	
	Let \( W \) be an \( n \)-dimensional vector representing the processing times of the jobs in the set \( V \). Consider the agreement graph \( \overline{G} = (V, \overline{E}) \), defined as the complement of the conflict graph \( G \). By assigning the processing times in \( W \) as weights to the vertices of \( \overline{G} \), we obtain the weighted graph \( \overline{G}_w = (V, \overline{E}, W) \). Any set of conflicting jobs corresponds to an independent set in \( \overline{G}_w \), and the maximum-weight independent set thus provides a valid lower bound on the makespan. However, finding such a set is an NP-hard problem~\citep{GJ79}. To approximate this bound, we apply three greedy algorithms---\texttt{GWMIN} and \texttt{GWMIN2}---from~\cite{SMK03} to \( \overline{G}_w \). The resulting lower bounds are denoted by \( LB_2 \) and \( LB_3 \) respectively.
	
	\subsection{Binary variables relaxation-based lower bound}
	This class of bounds is obtained by solving the mathematical formulations introduced in Section~\ref{sec:MILP} under a fixed time limit. For each formulation, we retrieve the best lower bound found by the solver within the allotted time. These bounds are at least as tight as those provided by the linear relaxations. We denote the resulting bounds as $LB_4$, $LB_5$, $LB_6$ and $LB_{7}$, corresponding to \mFA{}, \mFB{}, \mFC{}, and \mFD{}, respectively.
	
	\section{Genetic algorithms}\label{sec:GA}
	Genetic Algorithms (GAs), originally proposed by Holland~\citep{Holland1975}, are population-based metaheuristics that maintain a population of candidate solutions evolving over successive generations through genetic operations (selection, crossover, mutation, and replacement). Each solution is encoded as a chromosome (composed of genes) and evaluated using a fitness function, which guides the search toward high-quality solutions even in the presence of multiple local optima.
	
	In this section, we propose a GA for the problem $Jm,|,ConfG=(V,E),|,C_{\max}$. The algorithm begins by generating an initial population of size $N_p$, from which the top $N_e$ chromosomes are selected to form the elite set. At each iteration, the new population is initialized with the elite solutions. While its size remains below $N_p$, pairs of parent chromosomes are selected from the previous population and undergo crossover with probability $p_c$. Mutation is then applied to the resulting offspring--or to the parents when crossover is not performed--with probability $p_m$. The resulting chromosomes are added to the current population. Once the population reaches size $N_p$, the top $N_e$ chromosomes are updated as the elite set, and the process repeats until a stopping criterion is met. The algorithm returns the best solution found so far. A summary of this process is provided in Algorithm~\ref{GAalgo}, and the procedures used in each step are detailed in the following sections.
	
	\begin{algorithm}[H]
		\caption{General framework of the GA}\label{GAalgo}
		\begin{algorithmic}[1]
			\State Generate an initial population of size $N_p$ and select the top $N_e$ chromosomes as the elite set;
			\While{none of the stopping criteria are met}
			\State Initialize the current population with the elite set;
			\While{the size of the current population is less than $N_p$}
			\State Select two parents $P_1$ and $P_2$ from the previous population and initialize $C_1 \gets P_1$, $C_2 \gets P_2$;
			\If{\texttt{Rand} $< p_c$}
			\State Apply a crossover operator to $(P_1, P_2)$ to generate offspring $(C_1, C_2)$;
			\EndIf
			\If{\texttt{Rand} $< p_m$}
			\State Apply a mutation operator to $C_1$ and $C_2$;
			\EndIf
			\State Evaluate $C_1$ and $C_2$;
			\State Add $C_1$ and $C_2$ to the current population;
			\EndWhile
			\State Update the elite set with the best $N_e$ chromosomes from the current population;
			\EndWhile
		\end{algorithmic}
		\hspace*{\algorithmicindent}\textbf{Output}: the best makespan and its associated schedule.
	\end{algorithm}

	\subsection{Encoding and decoding}
	
	We use a permutation $\pi$ with repetition representation~\cite{bierwirth1995generalized} to encode solutions. In this representation, each chromosome is a permutation where every job index $j$ appears exactly as many times as the number of operations of that job $m_j$. The $i$-th occurrence of a job implicitly corresponds to its $i$-th operation, since during decoding we always schedule the next unscheduled operation of the job. This ensures that precedence constraints within each job are satisfied. Moreover, this representation is well suited for permutation-based genetic operators, as it guarantees feasibility throughout the evolutionary process. For example, for three jobs with 3, 2, and 4 operations respectively, a valid permutation is 
	$\pi = (1,2,1,3,2,3,1,3,3)$, where each job index appears exactly as many times as the number of its operations.

	To evaluate a chromosome, we decode its permutation into a feasible schedule using two decoding procedures that generate non-delay or active schedules, and the resulting makespan serves as the chromosome’s fitness value. An active schedule is one in which no operation can be started earlier without delaying another operation. For regular performance measures, including makespan, every optimal schedule is necessarily active. Non-delay schedules form a subset of active schedules in which no machine remains idle if an operation is available for processing. Although not all optimal schedules are non-delay, this restriction often produces schedules with good makespan performance while substantially reducing the search space~\cite{baker2018principles}.
	
	Let $\pi'$ initially be a permutation of the first operations of each job, ordered as they appear in $\pi$. In the following algorithms, $\pi'$ is updated at each iteration: when an operation is scheduled, the next operation of the same job is inserted into $\pi'$ in the same relative position as it appears in $\pi$.  
	
	Algorithm~\ref{active_schedule} constructs active schedules based on the method of~\citet{GT1960}. At each iteration, the first operation $O_{i'j'}$ in $\pi'$ with the earliest completion time is selected. Next, the first operation in $\pi'$ that conflicts with $O_{i'j'}$ (including $O_{i'j'}$ itself) and has a start time earlier than the expected completion time of $O_{i'j'}$ is scheduled. The start times of all conflicting operations in $\pi'$, including the operation of $J_{j'}$ just inserted, are adjusted to the completion time of $O_{i'j'}$. This adjustment ensures that scheduling operations at their earliest start times will not violate the conflict constraints.
	
	We can show that the schedule produced is active by using the same argument as in Lemma 1 of~\cite{tellache2023genetic}, where a similar algorithm was used for the OSC problem.
	
	\begin{algorithm}[H]
		\caption{Build active schedules by Giffler and Thompson mechanism~\cite{GT1960}}
		\label{active_schedule}
		\hspace*{\algorithmicindent} \textbf{Inputs}:  permutation $\pi$,  processing times $(p_{ij})_{1\leq i \leq m_j, 1 \leq j \leq n}$, adjacency matrix $A$ of $G$  
		\begin{algorithmic}[1]
			\State Initialize $\pi'$ with the first operations of each job, ordered as they appear in $\pi$;
			\State Initialize $\mathcal U $ to all the operations of $\pi$;
			\State Initialize the earliest starting times $(s_{ij})_{1\leq i \leq m_j, 1 \leq j \leq n}$ of the operations of $\mathcal U $ to zero;
			\While {$\pi' \neq \emptyset $} 
			\State Select an operation $O_{i'j'}$ of $\pi'$ with the minimum earliest completion time $s_{i'j'}+p_{i'j'}$; \label{ctactive1}
			\State Select the first operation $O_{ij}$ of $\pi'$ that is in conflict with $O_{i'j'}$ (including $O_{i'j'}$) such that $s_{ij} < s_{i'j'}+p_{i'j'}$; \label{ctactive2}
			\State $c_{ij}=s_{ij}+p_{ij}$;
			\State $\pi'=\pi' \setminus \{O_{ij}\}$, $\mathcal U=\mathcal U \setminus \{O_{ij}\}$;
			\State Insert the next operation of the job corresponding to $O_{ij}$ (if any) into $\pi'$ in the same relative position as it appears in $\pi$;
			\For{ $O_{i''j''} \in \mathcal U$ such that $O_{i''j''}\sim O_{ij}$} \label{feasibility}
			\If{$s_{i''j''} < s_{ij}+p_{ij}$ }
			\State $s_{i''j''}=s_{ij}+p_{ij}$;
			\EndIf
			\EndFor
			\EndWhile      
		\end{algorithmic}
		\hspace*{\algorithmicindent} \textbf{Outputs}: completion times $(c_{ij})_{1\leq i \leq m_j, 1 \leq j \leq n}$ and the makespan.
	\end{algorithm}

	The non-delay schedules are generated using Algorithm~\ref{non_delay_schedule}. Its structure is similar to Algorithm~\ref{heuristic_warmstart}, but here the priority rule is determined by the order of operations in the permutation~$\pi$. For completeness, the full version of the algorithm adapted to this permutation is provided in Algorithm~\ref{non_delay_schedule}. At each iteration, the algorithm selects the first operation in $\pi'$ with the smallest earliest start time, inserts it into the schedule, and then updates the earliest start times of the operations in $\pi$ that have not been scheduled yet. This update is performed in a manner similar to Algorithm~\ref{active_schedule}, ensuring that the conflict constraints remain satisfied.
	
	The schedule produced by this procedure is non-delay. This can be shown using the same argument as in the proof of Lemma 2 of~\cite{tellache2023genetic}, where a similar algorithm was used for the OSC problem.
	
	\begin{algorithm}[H]
		\caption{Build non-delay schedules}\label{non_delay_schedule} 
		\hspace*{\algorithmicindent} \textbf{Inputs}:  permutation $\pi$, processing times $(p_{ij})_{1\leq i \leq m_j, 1 \leq j \leq n}$, adjacency matrix $A$ of $G$ 
		\begin{algorithmic}[1]
			\State Initialize $\pi'$ with the first operations of each job, ordered as they appear in $\pi$;
			\State Initialize $\mathcal U $ to all the operations of $\pi$;
			\State Initialize the earliest starting times $(s_{ij})_{1\leq i \leq m_j, 1 \leq j \leq n}$ of the operations of $\pi'$ to zero;
			\While {$\pi' \neq \emptyset $} 
			\State Select the first operation $O_{ij}$ of $\pi'$ with the minimum earliest starting time $s_{ij}$;\label{CTnondelay}
			\State $c_{ij}=s_{ij}+p_{ij}$;
			\State $\pi'=\pi'\setminus \{O_{ij}\}$, $\mathcal U=\mathcal U \setminus \{O_{ij}\}$;
			\State Insert the next operation of the job corresponding to $O_{ij}$ (if any) into $\pi'$ in the same relative position as it appears in $\pi$;
			\For{ $O_{i'j'} \in \mathcal U $ such that $O_{i'j'}\sim O_{ij}$}\label{Fornondelay}
			\If{$s_{i'j'} <  s_{ij}+p_{ij}$ }
			\State $s_{i'j'}=s_{ij}+p_{ij}$;
			\EndIf
			\EndFor
			\EndWhile
		\end{algorithmic}
		\hspace*{\algorithmicindent} \textbf{Outputs}: completion times $(c_{ij})_{1\leq i \leq m, 1 \leq j \leq n}$ and the makespan.
	\end{algorithm}
	\subsection{Initial population generation}
	\label{meta_initialization}
	The initial population is either generated entirely at random or seeded with chromosomes constructed using the nine priority rules described in Section~\ref{sec:warm_start}. For each sorting, a chromosome is generated by representing every operation with the index of its corresponding job, and only unique chromosomes are retained in the population.
	
	\subsection{Selection, crossover, and mutation}
	The generation process of our algorithm relies on three main components: parent selection, crossover, and mutation. First, parent chromosomes are selected from the current population; then, offspring are produced from these parents through crossover and mutation. 
	
	To select parents for reproduction, we use the $n$-size tournament selection~\cite{pezzella2008genetic}. In this method, two tournaments are performed, each involving a subset of chromosomes chosen at random from the population. The winner of each tournament--that is, the chromosome with the best fitness--is then selected as a parent for crossover.
	
	Crossover combines two parent chromosomes with probability $p_c$ to generate offspring by exchanging genetic information. Its purpose is to inherit and combine advantageous traits from both parents, thereby increasing the likelihood of producing higher-quality solutions. A variety of crossover operators have been proposed for permutation-based encodings. In our algorithm, we use one of three operators: one-point, two-point, and mask-based crossover~\cite{REEVES1995}. In the one-point crossover, each parent is cut at a randomly selected position, and the segments beyond the cut are exchanged to form two offspring. In the two-point crossover, two cut points are selected at random, and the segment between these points is swapped between the parents. In the mask-based crossover, also known as uniform crossover or gene-wise recombination, a binary mask of length equal to the chromosome size is generated according to a Bernoulli distribution. For each gene position, the offspring inherits the gene from the first parent if the corresponding mask value is equal to one; otherwise, the gene is inherited from the second parent. When applied to the encoding adopted for the JSC--where each job identifier must appear in the chromosome a number of times equal to its total number of operations--this position-wise recombination may violate the required job-occurrence constraints. Therefore, a post-processing (repair) procedure (explained in Section~\ref{sec:repairprocedure}) is employed to restore feasibility by adjusting gene frequencies while preserving the relative ordering of operations as much as possible.
	
	Mutation is applied to each chromosome with probability $p_m$, introducing random changes to its genes to maintain diversity in the population and enable independent evolution. We use one of the following mutation operators: swap, which exchanges genes at two randomly selected positions, and move, which relocates a gene from one randomly chosen position to another~\cite{REEVES1995}.
	
	\subsection{Post-processing (repair) procedure} \label{sec:repairprocedure}
	
	The application of the mask-based crossover to our chromosome representation may produce infeasible offspring, in the sense that the number of occurrences of a job index $k$ in a chromosome differs from the required value $m_k$. To restore feasibility, a post-processing (repair) procedure is applied to each offspring chromosome before decoding and fitness evaluation.
	
	Let $\pi = (g_1,\ldots,g_{\sum_{j}m_j})$ denote an offspring chromosome. For each job $J_k$, $k \in \{1,\ldots,n\}$, its occurrence count is defined as $f_k(\pi) = | \{ l \in \{1,\ldots,\sum_{j}m_j\} \mid g_l = k \}|.$ The chromosome $\pi$ is feasible if and only if $f_k(\pi) = m_k$ for all jobs $J_k$. Otherwise, we define the sets of surplus and deficit jobs as
	$ J^{+} = \{ J_k \mid f_k(\pi) > m_k \},$ $J^{-} = \{ J_k \mid f_k(\pi) < m_k \}. $
	
	The repair procedure replaces surplus occurrences of job indices in $J^{+}$ with missing indices from $J^{-}$ until the feasibility condition $f_k(\pi) = m_k$ holds for all $J_k$. Replacements are performed sequentially and in a left-to-right manner, so as to preserve the relative ordering of unaffected genes as much as possible and to limit disruption of inherited genetic structures. After repair, each job index $k$ appears exactly $m_k$ times in the chromosome.
	
	\section{Computational experiments}\label{sec:CE}
	
	\label{instances_description} 
	This section presents the computational experiments conducted to assess the performance of the mathematical formulations, the lower bounds, and the genetic algorithm. All algorithms were implemented in C++. The mathematical models were solved with IBM ILOG CPLEX 22.1.1 with a one-hour time limit for each instance. The experiments were carried out on a machine equipped with an 11\textsuperscript{th}-generation Intel\textsuperscript{\textregistered} Core\textsuperscript{TM} i5-11500 @ 2.70 GHz × 12 processor and 32 GiB of RAM.
	
	The experiments were conducted on three instance sets. The first two were adapted from the benchmark instances of~\citet{taillard1993benchmarks} and~\citet{lawrence1984resouce}, originally proposed for the basic job-shop scheduling problem. The third set consists of randomly generated generalized job shop instances, in which a job may contain more than one operation assigned to the same machine.

	The Taillard benchmark set~\cite{taillard1993benchmarks} consists of 8 classes, each containing 10 instances. The classes correspond to the following problem sizes: $(n,m) \in \{(15,15),\allowbreak (20,15),\allowbreak (20,20), \allowbreak(30,15),\allowbreak (30,20),\allowbreak (50,15),\allowbreak (50,20),\allowbreak (100,20)\}$. In all instances, the number of operations in each job is equal to the number of machines. Moreover, the instances were constructed such that the total processing time is balanced across both machines and jobs.
	
	The Lawrence benchmark set~\cite{lawrence1984resouce} consists of 8 classes, each containing 5 instances. The classes correspond to the following problem sizes: $(n,m) \in \{(10,5),\allowbreak (15,5),\allowbreak (20,5),\allowbreak (10,10),\allowbreak (15,10), \allowbreak(20,10),\allowbreak (30,10), \allowbreak(15,15)\}$. In all instances, the number of operations in each job is equal to the number of machines, and each job visits every machine exactly once. 
	
	For each of these instances, conflict graphs were generated using the $G(n, p)$~\citet{ER59} algorithm. Given a set of $n$ vertices, this algorithm includes each of the $\binom{n}{2}$ possible edges with probability $p$. To analyze the impact of varying graphs densities, we considered three values of $p \in \{0.2, 0.5, 0.8\}$, corresponding to low, medium, and high density levels, respectively. For each instance from~\citet{taillard1993benchmarks} and~\citet{lawrence1984resouce} one conflict graph was generated for each value of $p$.
	
	For the generalized job shop instances, we considered the following classes of sizes $(n,m) \in \{(5,3),\allowbreak (8,5),\allowbreak (10,5),\allowbreak (20,10),\allowbreak (50,15),\allowbreak (100,20)\}$. For each instance, the number of operations per job was chosen uniformly at random between \(2\) and \(\tfrac{5}{3}m\). Each operation was then assigned to a machine selected uniformly at random, subject to the constraint \(\mu_{ij} \neq \mu_{(i-1)j}\) for $j=1,\ldots,n$ and \(i = 2, \ldots, m_j\), ensuring that consecutive operations of the same job are processed on distinct machines. Processing times were drawn uniformly at random from the intervals \([1,20]\), \([1,50]\), and \([1,100]\). The conflict graphs were constructed in the same manner as for the previous instance sets. For each combination of parameters, five instances were generated.
	
	In total, we obtained 240 instances derived from~\citet{taillard1993benchmarks}, 120 instances derived from \citet{lawrence1984resouce}, and 270 randomly generated instances.

	The remainder of this section is organized as follows. Section~\ref{secexp:milp} reports the performance of the mathematical formulations presented in Section~\ref{sec:MILP}. Section~\ref{secexp:lb} reports the performance of the lower bounds introduced in Section~\ref{sec:LB}. In Sections~\ref{sec:expoperatorsGA},~\ref{sec:exphybridGA}, and~\ref{sec:expparamGA}, we fine-tune the GA and assess the contribution of each component to its performance. Finally, Section~\ref{sec:finalGAresults} evaluates the overall performance of the GA variants against the constructive heuristics.

	\subsection{Performance of the mathematical models}\label{secexp:milp}
	This section evaluates the performance of the mathematical formulations presented in Section~\ref{sec:MILP}. As mentioned previously, a time limit of one hour was imposed for each run. Table~\ref{tab:milps} presents the results for the randomly generated instances with $(n,m)\in\{(5,3),(8,5)\}$. Larger instances are not reported, as several runs resulted in out-of-memory errors under the 32~GiB memory limit of the machine used for the experiments. For the time-indexed formulations, $T$ is set to the sum of the processing times of all operations.
	
	Table~\ref{tab:milps} reports, for each formulation, the number of instances solved to optimality ($\mathrm{Opt}$), the computational time in seconds ($\mathrm{Time}$), and the optimality gap at termination ($\mathrm{Gap}$) reported by Gurobi. The relative optimality gap is computed as $\frac{|UB-LB|}{|UB|}\times 100$,
	where $UB$ denotes the incumbent objective value and $LB$ the corresponding lower bound.
	
	The results in Table~\ref{tab:milps} show that $\mFB{}$ provides the best overall performance among the four formulations. It solves 49 out of 90 instances to optimality, with an average optimality gap of $0.276\%$. In comparison, $\mFA{}$, $\mFC{}$, and $\mFD{}$ solve 47, 23, and 24 instances to optimality, with average gaps of $0.286\%$, $0.427\%$, and $0.468\%$, respectively. Overall, the precedence-based formulations, $\mFA{}$ and $\mFB{}$, clearly outperform the time-indexed formulations, $\mFC{}$ and $\mFD{}$, particularly for instances with larger processing-time intervals. This can be partly attributed to the sensitivity of the time-indexed formulations to the value of $T$, which depends on the processing times. As $T$ increases, particularly for instances with large $p_{ij}$ values, the number of decision variables and constraints increases accordingly, resulting in larger models.
	
	Comparing the two time-indexed formulations, $\mFD{}$ performs better on small instances with low conflict densities, whereas $\mFC{}$ becomes better as the instance size and density increase. This suggests that, despite generating more constraints, the tighter pairwise conflict constraints in $\mFC{}$~\eqref{F2:conflict4} provide a stronger formulation than the degree-based aggregated constraints in $\mFD{}$~\eqref{F3:conflict}, particularly for larger and denser instances.
	
	Among the precedence-based formulations, $\mFB{}$ outperforms $\mFA{}$ on most instances. This indicates that the additional equality constraints in $\mFB{}$ strengthen the formulation, despite the introduction of symmetric variables.
	
	Overall, the results indicate that larger instance sizes and higher conflict densities generally lead to larger mathematical models and, consequently, increase the computational difficulty of the formulations. The processing-time range has a more pronounced effect on the time-indexed formulations, for which the number of variables and constraints grows with the discretized time horizon.
	
	\begin{table}[H]
		\centering
		\small
		\caption{Computational results of the mathematical formulations.}
		\label{tab:milps}
		\setlength{\tabcolsep}{3pt}
		\begin{tabular}{rrr rrr rrr rrr rrr}
			\toprule
			\multirow{2}{*}{Size} & \multirow{2}{*}{$p$} & \multirow{2}{*}{$p_{ij}\in$} 
			& \multicolumn{3}{c}{\mFA} 
			& \multicolumn{3}{c}{\mFB} 
			& \multicolumn{3}{c}{\mFC} 
			& \multicolumn{3}{c}{\mFD} \\
			\cmidrule(lr){4-6} \cmidrule(lr){7-9} \cmidrule(lr){10-12} \cmidrule(lr){13-15}
			&  &  & Opt & Time & Gap & Opt & Time & Gap & Opt & Time & Gap & Opt & Time & Gap \\
			\midrule
			
			\multirow{9}{*}{\shortstack{$n=5$ \\ $m=3$}} 
			& \multirow{3}{*}{0.2} 
			& [1,20]   
			& \textbf{5} & \textbf{0.082} & \textbf{0.00} 
			& \textbf{5} & 0.089 & \textbf{0.00} 
			& \textbf{5} & 70.06 & \textbf{0.00} 
			& \textbf{5} & 70.06 & \textbf{0.00} \\
			
			& & [1,50]  
			& \textbf{5} & 0.121 & \textbf{0.00} 
			& \textbf{5} & \textbf{0.088} & \textbf{0.00} 
			& \textbf{5} & 74.65 & \textbf{0.00} 
			& \textbf{5} & 74.65 & \textbf{0.00} \\
			
			& & [1,100] 
			& \textbf{5} & \textbf{0.060} & \textbf{0.00} 
			& \textbf{5} & \textbf{0.041} & \textbf{0.00} 
			& 3 & 2944.2 & 0.07 
			& 4 & 2944.2 & 0.03 \\
			\cmidrule(lr){2-15}
			
			& \multirow{3}{*}{0.5} 
			& [1,20]   
			& \textbf{5} & 4.62 & \textbf{0.00} 
			& \textbf{5} & \textbf{3.75} & \textbf{0.00} 
			& 1 & 3424.9 & 0.27 
			& 4 & 3424.9 & 0.07 \\
			
			& & [1,50]  
			& \textbf{5} & \textbf{20.53} & \textbf{0.00} 
			& \textbf{5} & 24.20 & \textbf{0.00} 
			& 2 & 4151.7 & 0.21 
			& 2 & 4151.7 & 0.21 \\
			
			& & [1,100] 
			& \textbf{5} & 55.74 & \textbf{0.00} 
			& \textbf{5} & \textbf{33.16} & \textbf{0.00} 
			& 1 & 4078.3 & 0.34 
			& 1 & 3375.6 & 0.33 \\
			\cmidrule(lr){2-15}
			
			& \multirow{3}{*}{0.8} 
			& [1,20]  
			& 2 & 2950.3 & 0.31 
			& \textbf{3} & \textbf{2573.7} & 0.21 
			& \textbf{3} & 2713.7 & \textbf{0.15} 
			& \textbf{3} & 2713.7 & \textbf{0.15} \\
			
			& & [1,50]  
			& \textbf{3} & \textbf{3385.8} & 0.16 
			& \textbf{3} & 3389.9 & \textbf{0.11} 
			& 0 & 6918.3 & 0.42 
			& 0 & 6918.3 & 0.45 \\
			
			& & [1,100] 
			& \textbf{2} & 2323.6 & 0.36 
			& \textbf{2} & \textbf{2066.2} & \textbf{0.35} 
			& 0 & 4467.2 & 0.53 
			& 0 & 4467.2 & 0.56 \\
			\midrule
			
			\multirow{9}{*}{\shortstack{$n=8$\\ $m=5$}} 
			& \multirow{3}{*}{0.2} 
			& [1,20]   
			& \textbf{5} & 1024.0 & \textbf{0.00} 
			& \textbf{5} & \textbf{987.8} & \textbf{0.00} 
			& 1 & 6993.8 & 0.26 
			& 0 & 6993.8 & 0.32 \\
			
			& & [1,50]  
			& \textbf{2} & 4244.3 & \textbf{0.26} 
			& \textbf{2} & \textbf{3975.1} & 0.27 
			& 0 & 4334.3 & 0.41 
			& 0 & 4334.3 & 0.41 \\
			
			& & [1,100] 
			& 2 & 4152.9 & 0.22 
			& \textbf{3} & 3843.7 & \textbf{0.17} 
			& 0 & \textbf{1800.0}& 0.69 
			& 0 & \textbf{1800.0} & 0.85 \\
			\cmidrule(lr){2-15}
			
			& \multirow{3}{*}{0.5} 
			& [1,20]   
			& \textbf{0} & 5067.2 & 0.60 
			& \textbf{0} & \textbf{4800.6} & 0.60 
			& \textbf{0} & 5294.2 & \textbf{0.45} 
			& \textbf{0} & 5294.2 & 0.53 \\
			
			& & [1,50]  
			& \textbf{1} & 4472.8 & 0.45 
			& \textbf{1} & 4309.7 & \textbf{0.44} 
			& 0 & \textbf{2488.0} & 0.50 
			& 0 & \textbf{2488.0} & 0.86 \\
			
			& & [1,100] 
			& \textbf{0} & 6679.0 & \textbf{0.53}
			& \textbf{0} & 6273.3 & 0.55 
			& \textbf{0} & \textbf{1800.9} & 1.00
			& \textbf{0} & \textbf{1800.9} & 1.00 \\
			\cmidrule(lr){2-15}
			
			& \multirow{3}{*}{0.8} 
			& [1,20]  
			& 0 & 3778.7 & 0.75 
			& 0 & \textbf{3179.0} & 0.75 
			& \textbf{2} & 4490.8 & \textbf{0.39} 
			& 0 & 4490.8 & 0.65 \\
			
			& & [1,50]  
			& \textbf{0} & 3375.0 & \textbf{0.76} 
			& \textbf{0} & 3104.2 & 0.77 
			& \textbf{0} & \textbf{1800.1} & 1.00
			& \textbf{0} & \textbf{1800.1} & 1.00\\
			
			& & [1,100] 
			& \textbf{0} & 3319.1 & \textbf{0.74} 
			& \textbf{0} & 2944.0 & \textbf{0.74} 
			& \textbf{0} & \textbf{1801.0} & 1.00
			& \textbf{0} & \textbf{1801.0} & 1.00 \\
			\bottomrule
		\end{tabular}
	\end{table}
	
	\subsection{Performance of the lower bounds}\label{secexp:lb}
	
	The results for the lower bounds presented in Section~\ref{sec:LB} are reported in Tables~\ref{tab:lawrence_bounds}--\ref{tab:gjsb_interval}. We also considered the LP relaxations of the MILP formulations. However, for medium and large size instances, the solver frequently ran out of memory, making it impossible to retrieve the corresponding bounds. Therefore, the MILP-based lower bounds are not considered in the following computational results.
	
	For each instance, the best lower bound is defined as $BestLB=\max_i\{LB_i\}$ over the three lower bounds. The deviation of a lower bound $LB_i$ from the best lower bound is computed as $((BestLB-LB_i)/BestLB)\times 100$. The tables report the mean (Mean) and standard deviation (Std) of the lower-bound values, the percentage of instances for which each method provides the best bound (\%Best), the average percentage deviation from the best lower bound (AvgDev (\%)), and the average computational time (AvgTime (s)). 
	
	\begin{table}[H]
		\centering
		\caption{Performance of the lower bounds on instances derived from Lawrence by conflict-graph density.}
		\label{tab:lawrence_bounds}
		\small
		\begin{tabular}{llrrrrr}
			\toprule
			\multicolumn{1}{c}{$p$} & \multicolumn{1}{c}{Lower bound} & \multicolumn{1}{c}{Mean} & \multicolumn{1}{c}{Std} & \multicolumn{1}{c}{\%Best} & \multicolumn{1}{c}{AvgDev (\%)} & \multicolumn{1}{c}{AvgTime (s)} \\
			\midrule
			\multirow{3}{*}{0.2} & $LB_1$ & 1038.95 & 355.94 & 52.5 & 14.54 & 0.00062 \\
			& $LB_2$ & 1124.30 & 506.81 & 50.0 & 12.58 & 0.00017 \\
			& $LB_3$ & 1112.08 & 508.70 & 45.0 & 13.64 & 0.00017 \\
			\midrule
			\multirow{3}{*}{0.5} & $LB_1$ & 1038.95 & 355.94 & 0.0 & 51.74 & 0.00060 \\
			& $LB_2$ & 2391.97 & 1017.79 & 95.0 & 0.24 & 0.00016 \\
			& $LB_3$ & 2382.38 & 1016.09 & 92.5 & 0.58 & 0.00015 \\
			\midrule
			\multirow{3}{*}{0.8} &  $LB_1$ & 1038.95 & 355.94 & 0.0 & 80.47 & 0.00065 \\
			& $LB_2$ & 6198.52 & 3166.45 & 90.0 & 0.08 & 0.00019 \\
			& $LB_3$ & 6202.95 & 3168.12 & 100.0 & 0.00 & 0.00017 \\
			\bottomrule
		\end{tabular}
	\end{table}
	
	\begin{table}[H]
		\centering
		\caption{Performance of the lower bounds on instances derived from Taillard by conflict-graph density.}
		\label{tab:taillard_bounds}
		\small
		\begin{tabular}{llrrrrr}
			\toprule
			\multicolumn{1}{c}{$p$} & \multicolumn{1}{c}{Lower bound} & \multicolumn{1}{c}{Mean} & \multicolumn{1}{c}{Std} & \multicolumn{1}{c}{\%Best} & \multicolumn{1}{c}{AvgDev (\%)} & \multicolumn{1}{c}{AvgTime (s)} \\
			\midrule
			\multirow{3}{*}{0.2} & $LB_1$ & 2227.09 & 1366.05 & 30.0 & 22.41 & 0.00082 \\
			& $LB_2$& 2386.29 & 684.61 & 67.5 & 8.51 & 0.00073 \\
			& $LB_3$ & 2383.74 & 694.95 & 66.2 & 8.79 & 0.00073 \\
			\midrule
			\multirow{3}{*}{0.5} & $LB_1$ & 2227.09 & 1366.05 & 0.0 & 65.33 & 0.00086 \\
			& $LB_2$ & 6013.19 & 1821.36 & 83.8 & 1.06 & 0.00077 \\
			& $LB_3$ & 6042.59 & 1875.75 & 88.8 & 0.79 & 0.00081 \\
			\midrule
			\multirow{3}{*}{0.8} & $LB_1$ & 2227.09 & 1366.05 & 0.0 & 90.29 & 0.00085 \\
			& $LB_2$ & 22670.21 & 12286.37 & 75.0 & 0.19 & 0.00077 \\
			& $LB_3$ & 22700.03 & 12324.53 & 92.5 & 0.12 & 0.00080 \\
			\bottomrule
		\end{tabular}
	\end{table}
	
	\begin{table}[H]
		\centering
		\caption{Performance of the lower bounds on the generalized job shop instances by conflict-graph density.}
		\label{tab:gjsb_density}
		\small
		\begin{tabular}{llrrrrr}
			\toprule
			\multicolumn{1}{c}{$p$} & \multicolumn{1}{c}{Lower bound} & \multicolumn{1}{c}{Mean} & \multicolumn{1}{c}{Std} & \multicolumn{1}{c}{\%Best} & \multicolumn{1}{c}{AvgDev (\%)} & \multicolumn{1}{c}{AvgTime (s)} \\
			\midrule
			\multirow{3}{*}{0.2} & $LB_1$ & 2095.15 & 2205.23 & 83.3 & 2.47 & 0.00077 \\
			& $LB_2$ & 1484.32 & 1323.04 & 15.0 & 21.98 & 0.00088 \\
			& $LB_3$ & 1463.33 & 1316.12 & 15.0 & 22.94 & 0.00082 \\
			\midrule
			\multirow{3}{*}{0.5} & $LB_1$ & 2150.95 & 2179.49 & 0.0 & 39.98 & 0.00081 \\
			& $LB_2$ & 3502.70 & 3386.73 & 86.7 & 1.22 & 0.00092 \\
			& $LB_3$ & 3508.92 & 3387.69 & 85.0 & 1.55 & 0.00087 \\
			\midrule
			\multirow{3}{*}{0.8} & $LB_1$ & 2127.42 & 2247.62 & 0.0 & 78.91 & 0.00081 \\
			& $LB_2$ & 11841.12 & 12967.79 & 66.7 & 0.49 & 0.00093 \\
			& $LB_3$ & 11772.80 & 12851.31 & 78.3 & 0.63 & 0.00077 \\
			\bottomrule
		\end{tabular}
	\end{table}
	
	As can be observed from Tables~\ref{tab:lawrence_bounds}--\ref{tab:gjsb_interval}, the performance of the lower bounds is strongly influenced by the density of the conflict graph. $LB_1$ is competitive for sparse conflict graphs, particularly on the randomly generated instances with $p=0.2$, but its relative quality deteriorates considerably as the density increases. This behaviour is expected since $LB_1$ does not take the conflict structure into account. In contrast, $LB_2$ and $LB_3$ provide substantially tighter bounds for medium and high density conflict graphs, with average deviations from the best lower bound generally below $1\%$ for $p=0.5$ and $p=0.8$. This confirms the importance of incorporating the conflict-graph structure when deriving lower bounds.
	
	\begin{table}[H]
		\centering
		\caption{Performance of the lower bounds by processing-time interval and conflict-graph density on the generalized job shop instances.}
		\label{tab:gjsb_interval}
		\small
		\begin{tabular}{lllrrrrr}
			\toprule
			\multicolumn{1}{c}{Interval} & \multicolumn{1}{c}{$p$} & \multicolumn{1}{c}{Lower bound} & \multicolumn{1}{c}{Mean} & \multicolumn{1}{c}{Std} & \multicolumn{1}{c}{\%Best} & \multicolumn{1}{c}{AvgDev (\%)} & \multicolumn{1}{c}{AvgTime (s)} \\
			\midrule
			\multirow{9}{*}{$[1,20]$} & \multirow{3}{*}{0.2} & $LB_1$ & 773.40 & 587.22 & 85.0 & 2.41 & 0.00077 \\
			&  & $LB_2$ & 561.05 & 354.34 & 15.0 & 22.66 & 0.00072 \\
			&  & $LB_3$ & 561.05 & 354.34 & 15.0 & 22.66 & 0.00082 \\
			\cmidrule{2-8}
			& \multirow{3}{*}{0.5} & $LB_1$ & 778.30 & 606.42 & 0.0 & 41.79 & 0.00082 \\
			&  & $LB_2$ & 1292.25 & 892.55 & 85.0 & 1.94 & 0.00088 \\
			&  & $LB_3$ & 1303.95 & 919.90 & 85.0 & 2.21 & 0.00084 \\
			\cmidrule{2-8}
			& \multirow{3}{*}{0.8} & $LB_1$ & 755.75 & 562.65 & 0.0 & 79.41 & 0.00079 \\
			&  & $LB_2$ & 4370.25 & 3621.77 & 65.0 & 0.47 & 0.00091 \\
			&  & $LB_3$ & 4352.85 & 3589.56 & 80.0 & 0.44 & 0.00081 \\
			\midrule
			\multirow{9}{*}{$[1,50]$} & \multirow{3}{*}{0.2} & $LB_1$ & 1836.30 & 1393.65 & 90.0 & 1.64 & 0.00073 \\
			&  & $LB_2$ & 1320.65 & 790.76 & 10.0 & 20.69 & 0.00084 \\
			&  & $LB_3$ & 1281.65 & 763.41 & 10.0 & 22.63 & 0.00083 \\
			\cmidrule{2-8}
			& \multirow{3}{*}{0.5} & $LB_1$ & 1912.85 & 1444.48 & 0.0 & 39.86 & 0.00073 \\
			&  & $LB_2$ & 3060.60 & 2092.82 & 85.0 & 1.54 & 0.00087 \\
			&  & $LB_3$ & 3140.40 & 2223.92 & 85.0 & 0.89 & 0.00079 \\
			\cmidrule{2-8}
			& \multirow{3}{*}{0.8} & $LB_1$ & 1891.85 & 1467.58 & 0.0 & 78.52 & 0.00077 \\
			&  & $LB_2$ & 10587.60 & 8618.15 & 80.0 & 0.13 & 0.00112 \\
			&  & $LB_3$ & 10527.85 & 8596.72 & 70.0 & 0.57 & 0.00078 \\
			\midrule
			\multirow{9}{*}{$[1,100]$} & \multirow{3}{*}{0.2} & $LB_1$ & 3675.75 & 2876.57 & 75.0 & 3.35 & 0.00082 \\
			&  & $LB_2$ & 2571.25 & 1586.08 & 20.0 & 22.59 & 0.00108 \\
			&  & $LB_3$ & 2547.30 & 1594.36 & 20.0 & 23.54 & 0.00081 \\
			\cmidrule{2-8}
			& \multirow{3}{*}{0.5} & $LB_1$ & 3761.70 & 2742.67 & 0.0 & 38.29 & 0.00088 \\
			&  & $LB_2$ & 6155.25 & 4205.40 & 90.0 & 0.18 & 0.00101 \\
			&  & $LB_3$ & 6082.40 & 4196.80 & 85.0 & 1.55 & 0.00096 \\
			\cmidrule{2-8}
			& \multirow{3}{*}{0.8} & $LB_1$ & 3734.65 & 2908.39 & 0.0 & 78.81 & 0.00088 \\
			&  & $LB_2$ & 20565.50 & 17154.48 & 55.0 & 0.87 & 0.00075 \\
			&  & $LB_3$ & 20437.70 & 16954.28 & 85.0 & 0.88 & 0.00073 \\
			\bottomrule
		\end{tabular}
	\end{table}
	
	The results also show that $LB_2$ and $LB_3$ exhibit very similar performance overall, although $LB_3$ generally provides a slight advantage on the largest and densest instances. Overall, $LB_2$ and $LB_3$ yield the best bound on $72.31\%$ and $75.09\%$ of the instances, respectively. Their average deviations from the best lower bound are $4.95\%$ for $LB_2$ and $5.30\%$ for $LB_3$. Both bounds rely on the same greedy principle: a vertex is selected according to a specific rule, after which the selected vertex and its neighbors are removed from the graph. This process is repeated until no vertex remains, and the selected vertices form the independent set used to compute the lower bound. The difference between $LB_2$ and $LB_3$ lies in the vertex-selection rule. In $LB_2$, the selected vertex maximizes the function $w(v)/(d_{G_i}(v)+1)$, where $w(v)$ is the weight of vertex $v$ and $d_{G_i}(v)$ is its degree in the graph $G_i$ at the $i^{\textrm{th}}$ iteration. In contrast, in $LB_3$, the selected vertex maximizes the function $w(v)/(\sum_{u \in N^{+}_{G_i}(v)}w(u))$, where $N^{+}_{G_i}(v)$ denotes the neighborhood of $v$ including $v$ itself.
	
	\subsection{Sensitivity analysis of core genetic algorithm operators} \label{sec:expoperatorsGA}
	We first evaluate the main components of the proposed GA, namely the schedule evaluation strategy (active and non-delay schedules), the crossover operator (one-point, two-point, and mask crossover), and the mutation operator (swap and insertion). The GA variants use a population size of $N_p=100$, two-point crossover with probability $p_c=0.85$, swap mutation with probability $p_m=0.1$, tournament selection, and a maximum of 1000 iterations. The results are summarized in Tables~\ref{tab:evaluation_type},~\ref{tab:crossover}, and~\ref{tab:best_config}.
	
	\subsubsection{Impact of the evaluation strategy}
	Table~\ref{tab:evaluation_type} compares the active and non-delay schedule evaluation strategies across the three instance sets and conflict graph densities. 
	
	\begin{table}[H]
		\centering
		\begin{threeparttable}
			\caption{Comparison of active and non-delay schedule evaluation on the considered instance sets.}
			\label{tab:evaluation_type}
			\begin{tabular}{llrrrrrr}
				\toprule
				& & \multicolumn{3}{c}{Active schedules (Algorithm~\ref{active_schedule})} 
				& \multicolumn{3}{c}{Non-delay schedules (Algorithm~\ref{non_delay_schedule})} \\
				\cmidrule(lr){3-5} \cmidrule(lr){6-8}
				Instance set & $p$ & Mean & Std & Time (s) & Mean & Std & Time (s) \\
				\midrule
				\multirow{3}{*}{Lawrence-derived}
				& 0.2 & 1392.4 & 724.3 & 8.55 & \textbf{1388.9} & 721.0 & 4.37 \\
				& 0.5 & 2602.0 & 1443.8 & 9.96 & \textbf{2568.7} & 1429.0 & 4.56 \\
				& 0.8 & \textbf{5835.4} & 3400.6 & 8.11 & 5835.7 & 3400.9 & 4.38 \\
				\midrule
				\multirow{3}{*}{Taillard-derived}
				& 0.2 & 3652.4 & 1580.2 & 117.94 & \textbf{3568.2} & 1499.1 & 54.31 \\
				& 0.5 & 8750.3 & 4892.9 & 142.24 & \textbf{8378.3} & 4622.6 & 84.62 \\
				& 0.8 & 24435.7 & 13762.8 & 79.24 & \textbf{24416.2} & 13712.4 & 45.37 \\
				\midrule
				\multirow{3}{*}{Generalized job shop}
				& 0.2 & 953.7 & 720.6 & 219.82 & \textbf{923.6} & 688.5 & 103.97 \\
				& 0.5 & 2200.6 & 1938.2 & 275.14 & \textbf{2095.6} & 1821.5 & 165.74 \\
				& 0.8 & 5238.2 & 4581.8 & 165.81 & \textbf{5186.7} & 4506.1 & 93.47 \\
				\bottomrule
			\end{tabular}
		\end{threeparttable}
	\end{table}
	
	For the Lawrence-derived instances, the two strategies provide comparable solution quality, with no statistically significant difference ($p\text{-value}=0.5476$). The non-delay evaluation yields slightly lower average objective values for $p=0.2$ and $p=0.5$, while the active evaluation is marginally better for $p=0.8$. However, the non-delay evaluation consistently requires substantially less computational time, approximately half that of the active evaluation.
	
	A similar pattern is observed for the Taillard-derived instances. The non-delay evaluation achieves lower average objective values at all density levels, although the difference is not statistically significant ($p\text{-value}=0.1529$). In contrast, its computational time is reduced by approximately 32--43\% compared with the active evaluation.
	
	For the generalized job shop instances, the non-delay evaluation consistently outperforms the active evaluation in terms of average objective value, with an average improvement of $2.22\%$ across the three density levels. This difference is statistically significant ($p\text{-value}<0.05$). Moreover, the non-delay evaluation reduces computational time by approximately 45--55\%.
	
	Overall, the results indicate that non-delay evaluation provides solution quality comparable to that of active evaluation for the Lawrence- and Taillard-derived instances, while requiring substantially less computational effort. For the generalized job shop instances, it also yields a statistically significant improvement in solution quality. These observations are consistent with the established properties of active and non-delay schedules in the scheduling literature. Non-delay schedules tend to produce tighter schedules by preventing machine idle time when an operation is available. However, restricting the search to non-delay schedules may exclude optimal solutions, whereas for the makespan objective, every optimal schedule is active. Thus, the use of non-delay evaluation represents a trade-off between potentially restricting the solution space and reducing computational effort. 
	
	\subsubsection{Impact of crossover operators}
	
	Table~\ref{tab:crossover} compares the performance of the three crossover operators on the considered instance sets and conflict graph densities.
	\begin{table}[H]
		\centering
		\begin{threeparttable}
			\caption{Comparison of crossover operators on the considered instance sets.}
			\label{tab:crossover}
			\begin{tabular}{llrrrrrr}
				\toprule
				& & \multicolumn{2}{c}{One-Point} & \multicolumn{2}{c}{Two-Point} & \multicolumn{2}{c}{Mask} \\
				\cmidrule(lr){3-4} \cmidrule(lr){5-6} \cmidrule(lr){7-8}
				Instance set & $p$ & Mean & Std & Mean & Std & Mean & Std \\
				\midrule
				\multirow{3}{*}{Lawrence-derived}
				& 0.2 & 1395.1 & 724.1 & 1392.2 & 722.7 & \textbf{1384.7} & 722.6 \\
				& 0.5 & 2589.2 & 1440.6 & 2587.2 & 1438.8 & \textbf{2579.7} & 1432.7 \\
				& 0.8 & 5835.6 & 3402.8 & 5835.5 & 3402.8 & \textbf{5835.5} & 3402.7 \\
				\midrule
				\multirow{3}{*}{Taillard-derived}
				& 0.2 & 3601.3 & 1522.7 & \textbf{3597.8} & 1526.4 & 3631.9 & 1574.8 \\
				& 0.5 & 8540.5 & 4726.5 & \textbf{8525.6} & 4726.1 & 8626.8 & 4843.2 \\
				& 0.8 & 24431.7 & 13741.5 & \textbf{24417.4} & 13719.0 & 24428.8 & 13773.8 \\
				\midrule
				\multirow{3}{*}{Generalized job shop}
				& 0.2 & 934.9 & 700.5 & \textbf{932.6} & 698.7 & 948.5 & 716.1 \\
				& 0.5 & 2140.7 & 1872.7 & \textbf{2137.4} & 1870.6 & 2166.2 & 1902.9 \\
				& 0.8 & 5208.8 & 4536.9 & \textbf{5202.3} & 4533.0 & 5226.3 & 4567.2 \\
				\bottomrule
			\end{tabular}
		\end{threeparttable}
	\end{table}
	
	For the Lawrence-derived instances, the mask crossover consistently achieves the lowest mean objective value across all density levels. Although the differences are relatively small, the results suggest that preserving distributed positional information through mask-based recombination may be slightly better suited to the structured characteristics of these instances. For the Taillard-derived  and generalized job shop instances, the two-point crossover consistently provides the lowest mean objective value across all density levels. The improvement over the other operators is modest but stable, suggesting that exchanging contiguous chromosome segments may better preserve high-quality building blocks while maintaining sufficient diversity during the search. Overall, the three crossover operators tend to exhibit similar performance across the different instance sets and density levels. Although small differences can be observed for specific instances, none of the operators shows a substantial or consistent advantage over the others.
	
	\subsubsection{Selection of GA configuration}
	
	The two mutation operators, swap and insertion, were also compared, with no statistically significant difference observed between them, indicating comparable performance. Table~\ref{tab:best_config} summarizes the best-performing combination of the three GA components studied in this section for each instance set and conflict graph density, based on the average objective value.
	
	\begin{table}[H]
		\centering
		\begin{threeparttable}
			\caption{Best GA configurations for the considered instance sets at different conflict-graph densities.}
			\label{tab:best_config}
			\begin{tabular}{llrrrr}
				\toprule
				Instance set & $p$ & Evaluation & Crossover & Mutation & Mean \\
				\midrule
				\multirow{3}{*}{Lawrence-derived}
				& 0.2 & Non-delay & Mask & Swap & 1381.2 \\
				& 0.5 & Non-delay & Mask & Insertion & 2562.0 \\
				& 0.8 & Active & Two-point & Insertion & 5835.1 \\
				\midrule
				\multirow{3}{*}{Taillard-derived}
				& 0.2 & Non-delay & Two-point & Insertion & 3552.4 \\
				& 0.5 & Non-delay & Two-point & Insertion & 8346.7 \\
				& 0.8 & Non-delay & Two-point & Swap & 24407.4 \\
				\midrule
				\multirow{3}{*}{Generalized job shop}
				& 0.2 & Non-delay & Two-point & Insertion & 916.7 \\
				& 0.5 & Non-delay & Two-point & Insertion & 2084.1 \\
				& 0.8 & Non-delay & Two-point & Insertion & 5175.3 \\
				\bottomrule
			\end{tabular}
		\end{threeparttable}
	\end{table}
	
	The selected configurations show a consistent preference for non-delay evaluation and, to a lesser extent, for the two-point crossover. In particular, non-delay evaluation is selected for eight of the nine instance-set--density combinations, while two-point crossover is selected for six. The mask crossover is selected for the remaining Lawrence-derived configurations, whereas the active evaluation is selected only once. These choices are broadly consistent with the results of the preceding sensitivity analyses.
	
	\subsection{Impact of hybrid evaluation strategy on GA performance} \label{sec:exphybridGA}
	The previous experiments showed that the non-delay evaluation generally provides the best trade-off between solution quality and computational efficiency. However, limiting the search to non-delay schedules may prevent the algorithm from reaching optimal schedules. To overcome this limitation, we consider a hybrid evaluation strategy that combines active and non-delay evaluations, thereby expanding the set of schedules explored by the GA.

	\begin{table}[H]
		\centering
		\caption{Impact of the hybrid evaluation strategy on the GA performance.}
		\label{tab:summary_all}
		\resizebox{\textwidth}{!}{%
			\begin{tabular}{llrrrrrrr}
				\toprule
				Instance set & $p$ & $p_{\mathrm{active}}=0.0$ & $p_{\mathrm{active}}=0.1$ &
				$p_{\mathrm{active}}=0.2$ & $p_{\mathrm{active}}=0.5$ &
				$p_{\mathrm{active}}=0.8$ & $p_{\mathrm{active}}=0.9$ &
				$p_{\mathrm{active}}=1.0$ \\
				\midrule
				
				Lawrence-derived & 0.2 & \textbf{1397.7} & 1388.4 & 1392.9 & 1386.3 & 1395.2 & 1394.7 & 1401.7 \\
				& 0.5 & 2576.8 & 2574.8 & 2575.1 & \textbf{2570.7} & 2596.6 & 2603.1 & 2605.7 \\
				& 0.8 & 5835.8 & 5835.7 & 5835.6 & 5835.6 & 5835.8 & \textbf{5835.5} & 5835.9 \\
				\midrule
				
				Taillard-derived & 0.2 & \textbf{3559.2} & 3562.6 & 3574.4 & 3562.2 & 3613.4 & 3629.3 & 3654.4 \\
				& 0.5 & 8368.9 & 8368.4 & 8378.2 & \textbf{8353.5} & 8524.9 & 8591.4 & 8745.7 \\
				& 0.8 & 24422.7 & 24424.8 & 24410.2 & \textbf{24393.4} & 24410.1 & 24436.7 & 24437.0 \\
				\midrule
				
				Generalized job shop & 0.2 & 922.4 & 923.2 & 923.3 & \textbf{918.4} & 935.4 & 940.9 & 950.1 \\
				& 0.5 & 2090.6 & 2090.3 & 2091.7 & \textbf{2086.3} & 2120.3 & 2130.9 & 2195.4 \\
				& 0.8 & 5179.6 & 5183.1 & 5179.7 & \textbf{5171.1} & 5203.1 & 5212.7 & 5233.3 \\
				\bottomrule
			\end{tabular}%
		}
	\end{table}
	
	Based on the previous analysis, the two-point crossover and insertion mutation are fixed for this stage. The remaining GA parameters are set to a population size of $N_p=100$, a crossover probability of $p_c=0.8$, a mutation probability of $p_m=0.2$, and a maximum of 1000 generations. The experiments are conducted on the same instance sets as in the previous stage.
	
	Let $p_{\mathrm{active}}$ denote the probability of evaluating an offspring using the active schedule. Thus, $p=0$ corresponds to the standard non-delay evaluation, whereas $p=1$ corresponds to the exclusive use of active evaluation. Intermediate values probabilistically combine the two evaluation strategies during the evolutionary process.
	
	Table~\ref{tab:summary_all} reports the average makespan for different values of $p_{\mathrm{active}}$. Overall, moderate values of $p_{\mathrm{active}}$ tend to provide the best performance, while relying heavily on active evaluation generally leads to poorer solution quality. In particular, $p_{\mathrm{active}}=0.5$ provides a competitive and often superior compromise across the three instance sets. The benefit of hybridization is particularly apparent for the generalized job shop instances, where $p_{\mathrm{active}}=0.5$ consistently provides the best results. For the Lawrence- and Taillard-derived instances, the differences between low and moderate values of $p_{\mathrm{active}}$ are generally small. This may indicate that instances allowing multiple operations per machine are more likely to benefit from expanding the search beyond non-delay schedules, as optimal schedules are not necessarily non-delay. Based on these observations, we use $p_{\mathrm{active}}=0.5$ in the subsequent experiments.
	
	\subsection{Tuning the numerical GA parameters} \label{sec:expparamGA}
	In this section, we tune the numerical parameters of the genetic algorithm, namely the population size, mutation rate, and crossover rate. Based on the previous experiments, the search operators and evaluation strategy are fixed throughout this stage. Specifically, we use the two-point crossover, insertion mutation, and hybrid evaluation strategy with the probability of active evaluation set to $p_{\mathrm{active}}=0.5$, which provided the best overall performance. The experiments are conducted on the same instance sets as those considered in the previous stages.

	\begin{table}[H]
		\centering
		\caption{Impact of population size on the GA performance.}
		\label{tab:pop_size_effect}
		\begin{tabular}{lrrrrr}
			\toprule
			Instance set & $N_p=20$ & $N_p=50$ & $N_p=100$ & $N_p=200$ & $N_p=400$ \\
			\midrule
			Lawrence-derived & 3291.45 & 3280.41 & 3273.09 & 3266.27 & 3257.84 \\
			Taillard-derived & 12222.78 & 12174.01 & 12131.06 & 12092.53 & 12049.58 \\
			Generalized job shop & 2758.22 & 2744.90 & 2733.33 & 2721.30 & 2708.71 \\
			\bottomrule
		\end{tabular}
	\end{table}
	
	\begin{table}[H]
		\centering
		\caption{Impact of mutation probability on the GA performance.}
		\label{tab:mut_rate_effect}
		\begin{tabular}{lrrrr}
			\toprule
			Instance set & $p_m=0.05$ & $p_m=0.1$ & $p_m=0.2$ & $p_m=0.5$ \\
			\midrule
			Lawrence-derived & 3277.51 & 3274.18 & 3269.70 & 3264.15 \\
			Taillard-derived & 12162.92 & 12144.59 & 12117.35 & 12090.07 \\
			Generalized job shop & 2743.71 & 2737.39 & 2728.24 & 2718.60 \\
			\bottomrule
		\end{tabular}
	\end{table}
	
	\begin{table}[H]
		\centering
		\caption{Impact of crossover probability on the GA performance.}
		\label{tab:cross_rate_effect}
		\begin{tabular}{lrrr}
			\toprule
			Instance set & $p_c=0.7$ & $p_c=0.8$ & $p_c=0.9$ \\
			\midrule
			Lawrence-derived & 3273.16 & 3271.90 & 3271.17 \\
			Taillard-derived & 12137.15 & 12131.17 & 12125.47 \\
			Generalized job shop & 2735.79 & 2733.08 & 2730.20 \\
			\bottomrule
		\end{tabular}
	\end{table}
	
	\begin{table}[H]
		\centering
		\caption{Kruskal--Wallis test results for the significance of the numerical GA parameters.}
		\label{tab:kruskal_wallis}
		\small
		\begin{tabular}{llcccc}
			\toprule
			Instance set & $p$ & $H$-statistic & $p$-value & $\eta^2$ & Significant \\
			\midrule
			
			\multicolumn{6}{l}{\textit{Population size ($N_p$)}} \\
			\midrule
			Overall & -- & 7.68 & 0.1042 & 0.0001 & No \\
			
			\multirow{3}{*}{Lawrence-derived}
			& 0.2 & 72.53 & $<0.0001$ & 0.0163 & Yes \\
			& 0.5 & 30.51 & $<0.0001$ & 0.0063 & Yes \\
			& 0.8 & 0.28 & 0.9908 & 0.0000 & No \\
			
			\multirow{3}{*}{Taillard-derived}
			& 0.2 & 39.77 & $<0.0001$ & 0.0075 & Yes \\
			& 0.5 & 47.05 & $<0.0001$ & 0.0090 & Yes \\
			& 0.8 & 7.63 & 0.1059 & 0.0008 & No \\
			
			\multirow{3}{*}{Generalized job shop}
			& 0.2 & 39.59 & $<0.0001$ & 0.0049 & Yes \\
			& 0.5 & 57.18 & $<0.0001$ & 0.0074 & Yes \\
			& 0.8 & 23.74 & $<0.0001$ & 0.0027 & Yes \\
			
			\midrule
			\multicolumn{6}{l}{\textit{Mutation probability ($p_m$)}} \\
			\midrule
			Overall & -- & 1.66 & 0.6465 & 0.0000 & No \\
			
			\multirow{3}{*}{Lawrence-derived}
			& 0.2 & 10.32 & 0.0160 & 0.0017 & Yes \\
			& 0.5 & 7.30 & 0.0629 & 0.0010 & No \\
			& 0.8 & 0.03 & 0.9984 & 0.0000 & No \\
			
			\multirow{3}{*}{Taillard-derived}
			& 0.2 & 6.82 & 0.0780 & 0.0008 & No \\
			& 0.5 & 9.56 & 0.0227 & 0.0014 & Yes \\
			& 0.8 & 1.25 & 0.7408 & 0.0000 & No \\
			
			\multirow{3}{*}{Generalized job shop}
			& 0.2 & 10.46 & 0.0151 & 0.0010 & Yes \\
			& 0.5 & 17.08 & 0.0007 & 0.0020 & Yes \\
			& 0.8 & 4.88 & 0.1806 & 0.0003 & No \\
			
			\midrule
			\multicolumn{6}{l}{\textit{Crossover probability ($p_c$)}} \\
			\midrule
			Overall & -- & 0.06 & 0.9680 & 0.0000 & No \\
			
			\multirow{3}{*}{Lawrence-derived}
			& 0.2 & --- & $>0.05$ & 0.0000 & No \\
			& 0.5 & --- & $>0.05$ & 0.0000 & No \\
			& 0.8 & --- & $>0.05$ & 0.0000 & No \\
			
			\multirow{3}{*}{Taillard-derived}
			& 0.2 & --- & $>0.05$ & 0.0000 & No \\
			& 0.5 & --- & $>0.05$ & 0.0000 & No \\
			& 0.8 & --- & $>0.05$ & 0.0000 & No \\
			
			\multirow{3}{*}{Generalized job shop}
			& 0.2 & --- & $>0.05$ & 0.0000 & No \\
			& 0.5 & --- & $>0.05$ & 0.0000 & No \\
			& 0.8 & --- & $>0.05$ & 0.0000 & No \\
			
			\bottomrule
		\end{tabular}
	\end{table}
	
	The population size $N_p$ is varied over $\{20,50,100,200,400\}$, the mutation probability $p_m$ over $\{0.05,0.10,0.20,0.50\}$, and the crossover probability $p_c$ over $\{0.70,0.80,0.90\}$. Each parameter is evaluated independently while keeping the remaining parameters fixed.
	
	For the population size, table~\ref{tab:pop_size_effect} reports the average objective values obtained for the considered population sizes. A consistent improvement is observed as the population size increases from 20 to 400 across all three instance sets, with $N_p=400$ yielding the lowest average objective value in each case. The overall Kruskal--Wallis test does not indicate a statistically significant effect of population size ($p\text{-value}=0.1042$). However, significant differences are observed for several instance set--density combinations (Table~\ref{tab:kruskal_wallis}). In particular, the Lawrence instances with $p=0.2$ and $0.5$, the Taillard instances with $p=0.2$ and $0.5$, and all generalized job shop instances except $p=0.8$ exhibit statistically significant differences. Nevertheless, the corresponding effect sizes are consistently small ($\eta^2<0.02$), indicating that the observed improvements associated with larger populations are relatively modest. Based on these results, $N_p=400$ is selected for the final configuration.
	
	For the mutation probability, Table~\ref{tab:mut_rate_effect} reports the average objective values obtained for the considered mutation rates. Across all three instance sets, the average objective value decreases as the mutation rate increases, with $p_m=0.5$ consistently yielding the best average performance. The Kruskal--Wallis analysis, however, indicates that the effect is statistically significant only for a subset of instance set--density combinations (Table~\ref{tab:kruskal_wallis}). Significant differences are observed for the Lawrence instances with $p=0.2$, the Taillard instances with $p=0.5$, and the generalized job shop instances with $p=0.2$ and $0.5$. No statistically significant differences are observed for the remaining configurations. Moreover, the corresponding effect sizes are consistently very small, indicating that the practical impact of the mutation rate on solution quality is limited despite the observed numerical trend. Based on these results, $p_m=0.5$ is selected for the final configuration.

	For the crossover probability, Table~\ref{tab:cross_rate_effect} reports the average objective values obtained for the considered crossover rates. Increasing the crossover rate from $p_c=0.7$ to $0.9$ results in a slight and consistent improvement across all three instance sets, with $p_c=0.9$ yielding the lowest average objective value in each case. However, the magnitude of these differences is small. This is confirmed by the Kruskal--Wallis analysis, which detects no statistically significant differences either overall or for any of the individual instance set--density combinations ($p\text{-value}>0.05$ in all cases; Table~\ref{tab:kruskal_wallis}). Thus, the crossover rate has a limited practical impact within the investigated range. Based on its consistently best average performance, $p_c=0.9$ is selected for the final configuration.

	\subsection{Final computational evaluation} \label{sec:finalGAresults}
	Following the parameter tuning described above, the final configuration of the proposed GA is evaluated on the complete set of instances. The GA uses the two-point crossover, insertion mutation, tournament selection with a tournament size of five, a population size of $400$, a crossover rate of $p_c=0.9$, a mutation rate of $p_m=0.5$, and the hybrid evaluation strategy with $p_{\mathrm{active}}=0.5$. The evolutionary process terminates after 10,000 generations or earlier if no improvement is observed for 100 consecutive generations. Each instance is solved 10 times, and the reported results are averaged over these runs.
	
	The objectives of this section is are twofold: (i) to assess the performance of the GA against the constructive heuristics developed in Section~\ref{sec:warm_start}, and (ii) to evaluate the benefit of heuristic seeding in the evolutionary search.
	
	The results are reported in Tables~\ref{tab:main_Lawrence},~\ref{tab:main_Taillard}, and~\ref{tab:main_Generalized}. For the constructive heuristics, we report the results of the best-performing heuristic identified in Section~\ref{sec:warm_start}, denoted by $H$. For the GA, we consider the standard GA, initialized entirely with randomly generated chromosomes and denoted by $GA$, and a GA whose initial population is seeded with heuristic solutions, denoted by $GA_h$. Its initial population includes the chromosomes generated by the heuristics in Section~\ref{sec:warm_start}, with the remaining chromosomes generated randomly.
	
	Solution quality is evaluated using the relative percentage deviation (RPD) from the best available lower bound $BestLB$: $ \mathrm{RPD}=\frac{C_{\max}-BestLB}{BestLB}\times100$, standard deviation in percentage ($\sigma$), and mean CPU time in seconds (Time).
	
	The results reported in Tables~\ref{tab:main_Lawrence},~\ref{tab:main_Taillard}, and~\ref{tab:main_Generalized} show that the genetic algorithms generally outperform the constructive heuristic. For the Lawrence (resp. Taillard) instances, the average RPD decreases from $31.78\%$ (resp. $34.40\%$) to $17.42\%$ (resp. $22.64\%$) at $p=0.2$, from $19.08\%$ (resp. $40.71\%$) to approximately $9\%$ (resp. $29.43\%$) at $p=0.5$, and from $1.17\%$ (resp. $1.77\%$) to $0.90\%$ (resp. $1.21\%$) at $p=0.8$. For the Generalized job shop instances, the improvement is more modest, with the average RPD decreasing from $59.89\%$ to $56.40\%$ at $p=0.2$, from $48.10\%$ to $42.30\%$ at $p=0.5$, and from $9.92\%$ to $8.74\%$ at $p=0.8$.
	
	\begin{table}[H]
		\centering
		\small
		\setlength{\tabcolsep}{3.5pt}
		\caption{Results on Lawrence-derived instances.}
		\label{tab:main_Lawrence}
		\begin{tabular}{c r r r rrr rrr rrr}
			\toprule
			& & & & \multicolumn{3}{c}{$H$} & \multicolumn{3}{c}{$GA$} & \multicolumn{3}{c}{$GA_h$} \\
			\cmidrule(lr){5-7} \cmidrule(lr){8-10} \cmidrule(lr){11-13}
			$p$ & $n$ & $m$ & \#Inst & RPD & $\sigma$ & Time & RPD & $\sigma$ & Time & RPD & $\sigma$ & Time \\
			\midrule
			0.2 & 10 & 5 & 5 & 21.37 & 16.26 & 0.00 & \textbf{7.56} & 9.63 & 3.71 & 7.77 & 9.87 & 3.86 \\
			& 10 & 10 & 5 & 29.02 & 35.01 & 0.00 & 22.55 & 30.99 & 2.67 & \textbf{22.46} & 30.88 & 2.53 \\
			& 15 & 5 & 5 & 59.14 & 36.76 & 0.00 & 32.93 & 27.72 & 3.13 & \textbf{32.63} & 27.50 & 3.20 \\
			& 15 & 10 & 5 & 33.66 & 5.59 & 0.00 & \textbf{17.20} & 3.61 & 16.46 & 17.21 & 3.68 & 15.37 \\
			& 15 & 15 & 5 & 18.58 & 15.89 & 0.00 & \textbf{6.07} & 7.23 & 21.87 & 6.31 & 7.51 & 23.47 \\
			& 20 & 5 & 5 & 21.90 & 27.12 & 0.00 & \textbf{13.38} & 24.84 & 8.09 & 13.50 & 25.10 & 8.37 \\
			& 20 & 10 & 5 & 31.43 & 22.38 & 0.00 & \textbf{14.81} & 11.97 & 29.62 & 14.89 & 11.87 & 32.27 \\
			& 30 & 10 & 5 & 39.19 & 16.09 & 0.00 & \textbf{24.82} & 17.54 & 70.03 & 25.19 & 17.08 & 71.74 \\
			\addlinespace[4pt]
			0.5 & 10 & 5 & 5 & 16.89 & 8.68 & 0.00 & 6.13 & 8.23 & 3.52 & \textbf{6.06} & 8.18 & 3.99 \\
			& 10 & 10 & 5 & 14.88 & 3.41 & 0.00 & 4.51 & 4.08 & 7.27 & \textbf{4.49} & 4.00 & 7.21 \\
			& 15 & 5 & 5 & 16.08 & 4.90 & 0.00 & 4.50 & 5.26 & 9.35 & \textbf{4.44} & 5.16 & 9.41 \\
			& 15 & 10 & 5 & 9.34 & 2.97 & 0.00 & \textbf{1.12} & 1.22 & 19.59 & 1.18 & 1.10 & 17.67 \\
			& 15 & 15 & 5 & 14.46 & 7.26 & 0.00 & \textbf{5.37} & 5.87 & 29.40 & 5.50 & 5.97 & 31.62 \\
			& 20 & 5 & 5 & 22.40 & 10.85 & 0.00 & 13.49 & 8.17 & 20.08 & \textbf{13.42} & 8.15 & 20.37 \\
			& 20 & 10 & 5 & 23.96 & 10.15 & 0.00 & \textbf{15.00} & 10.19 & 45.49 & 15.13 & 10.32 & 44.74 \\
			& 30 & 10 & 5 & 34.65 & 8.28 & 0.00 & 21.92 & 9.44 & 108.27 & \textbf{21.88} & 9.44 & 109.55 \\
			\addlinespace[4pt]
			0.8 & 10 & 5 & 5 & 0.91 & 0.96 & 0.00 & \textbf{0.70} & 0.76 & 2.55 & \textbf{0.70} & 0.76 & 2.56 \\
			& 10 & 10 & 5 & \textbf{0.69} & 0.91 & 0.00 & \textbf{0.69} & 0.91 & 4.16 & \textbf{0.69} & 0.91 & 4.17 \\
			& 15 & 5 & 5 & 0.94 & 0.96 & 0.00 & \textbf{0.67} & 0.69 & 4.48 & \textbf{0.67} & 0.69 & 4.48 \\
			& 15 & 10 & 5 & 1.14 & 1.48 & 0.00 & \textbf{1.00} & 1.41 & 7.28 & \textbf{1.00} & 1.41 & 7.34 \\
			& 15 & 15 & 5 & 0.42 & 0.63 & 0.00 & \textbf{0.30} & 0.37 & 9.93 & \textbf{0.30} & 0.37 & 10.03 \\
			& 20 & 5 & 5 & 1.62 & 0.68 & 0.00 & \textbf{1.42} & 0.89 & 11.49 & \textbf{1.42} & 0.89 & 11.54 \\
			& 20 & 10 & 5 & 1.73 & 1.53 & 0.00 & \textbf{1.11} & 0.87 & 15.34 & \textbf{1.11} & 0.87 & 15.61 \\
			& 30 & 10 & 5 & 1.89 & 0.99 & 0.00 & \textbf{1.31} & 0.73 & 41.06 & \textbf{1.32} & 0.73 & 39.92 \\
			\midrule
			0.2 & \textit{All} &  & 40 & 31.78 & 25.04 & 0.00 & \textbf{17.42} & 19.40 & 19.45 & 17.49 & 19.30 & 20.10 \\
			0.5 & \textit{All} &  & 40 & 19.08 & 10.09 & 0.00 & \textbf{9.00} & 9.29 & 30.37 & 9.01 & 9.29 & 30.57 \\
			0.8 & \textit{All} &  & 40 & 1.17 & 1.08 & 0.00 & \textbf{0.90} & 0.87 & 12.04 & \textbf{0.90} & 0.87 & 11.96 \\
			\bottomrule
		\end{tabular}
	\end{table}
	
	The statistical analysis confirms these differences (see Table~\ref{tab:stat_tests}). The Friedman test rejects the null hypothesis of equal performance for all benchmark--density combinations except the generalized job shop instances at \(p=0.2\). For the groups where the Friedman test is significant, the pairwise Wilcoxon signed-rank tests show that both \(GA\) and \(GA_h\) significantly outperform \(H\) after Bonferroni correction (\(\alpha=0.0167\)). Thus, the improvements observed in terms of average RPD are statistically significant for the Lawrence and Taillard instances across all densities, as well as for the Generalized instances at \(p=0.5\) and \(p=0.8\).
	
	Conflict graph density has a consistent effect on solution quality across the three approaches. Overall, increasing $p$ from 0.2 to 0.8 reduces the average RPD, except for the Taillard instances, for which the highest RPD is observed at $p=0.5$. This trend can be attributed, at least in part, to the tighter lower bounds provided by GWMIN and GWMIN2 for denser graphs, as well as to the more restricted feasible search space in these instances.

	The comparison between \(GA\) and \(GA_h\) reveals only marginal differences across the benchmark. The best RPD alternates between the two variants depending on the instance, while their average solution quality and computational times remain nearly identical across all conflict graph densities. This observation is also supported by the statistical analysis (Table~\ref{tab:stat_tests}): the pairwise Wilcoxon tests yield \(p\)-values above 0.05 for all benchmark--density combinations, indicating that the difference between \(GA\) and \(GA_h\) is not statistically significant. Although \(GA_h\) incorporates several heuristic-generated solutions into its initial population, this does not result in a systematic performance advantage over the random initialization of \(GA\). This suggests that the evolutionary operators and subsequent search process play a more significant role in determining the final solution quality than the initialization strategy.
	\begin{table}[H]
		\centering
		\small
		\setlength{\tabcolsep}{3.5pt}
		\caption{Results on Taillard-derived instances.}
		\label{tab:main_Taillard}
		\begin{tabular}{cccrrrr rrr rrr}
			\toprule
			& & & & \multicolumn{3}{c}{$H$} 
			& \multicolumn{3}{c}{$GA$} 
			& \multicolumn{3}{c}{$GA_H$} \\
			\cmidrule(lr){5-7} \cmidrule(lr){8-10} \cmidrule(lr){11-13}
			$p$ & $n$ & $m$ & \#Inst. 
			& RPD & $\sigma$ & Time 
			& RPD & $\sigma$ & Time 
			& RPD & $\sigma$ & Time \\
			\midrule
			0.2 & 15  & 15 & 10 & 21.54 & 13.23 & 0.00 & \textbf{11.60} & 12.05 & 19.74 & 11.71 & 12.17 & 19.74 \\
			& 20  & 15 & 10 & 31.23 & 12.26 & 0.00 & 18.80 & 9.27  & 44.27 & \textbf{18.60} & 9.25  & 45.86 \\
			& 20  & 20 & 10 & 27.01 & 16.62 & 0.00 & 9.91  & 7.97  & 57.09 & \textbf{9.82}  & 8.16  & 58.05 \\
			& 30  & 15 & 10 & 46.45 & 16.18 & 0.00 & 31.12 & 13.65 & 100.90 & \textbf{31.03} & 13.47 & 105.85 \\
			& 30  & 20 & 10 & 38.60 & 20.99 & 0.00 & \textbf{26.34} & 18.17 & 137.45 & 26.60 & 18.16 & 136.73 \\
			& 50  & 15 & 10 & 30.88 & 4.23  & 0.00 & 21.77 & 4.71  & 356.84 & \textbf{21.69} & 4.82  & 337.66 \\
			& 50  & 20 & 10 & 52.71 & 11.97 & 0.00 & 39.94 & 10.37 & 437.13 & \textbf{39.79} & 10.06 & 460.50 \\
			& 100 & 20 & 10 & 26.82 & 3.55  & 0.04 & 21.92 & 3.70  & 1740.13 & \textbf{21.90} & 3.60 & 1815.30 \\
			\addlinespace[4pt]
			0.5 & 15  & 15 & 10 & 13.60 & 9.65  & 0.00 & 7.10  & 7.63  & 28.24 & \textbf{7.03} & 7.63 & 25.58 \\
			& 20  & 15 & 10 & 17.65 & 13.24 & 0.00 & \textbf{8.71} & 9.34 & 63.95 & 8.76 & 9.34 & 63.54 \\
			& 20  & 20 & 10 & 19.21 & 7.54  & 0.00 & 9.77  & 6.92  & 86.01 & \textbf{9.75} & 6.95 & 89.97 \\
			& 30  & 15 & 10 & 34.60 & 11.01 & 0.00 & \textbf{21.15} & 9.96 & 162.26 & 21.17 & 10.02 & 161.94 \\
			& 30  & 20 & 10 & 26.94 & 7.87  & 0.00 & \textbf{17.42} & 6.60 & 240.52 & 17.43 & 6.55 & 226.46 \\
			& 50  & 15 & 10 & 53.91 & 7.47  & 0.01 & 39.74 & 7.04 & 534.12 & \textbf{39.73} & 7.17 & 528.35 \\
			& 50  & 20 & 10 & 49.20 & 7.85  & 0.02 & 36.20 & 7.87 & 727.21 & \textbf{36.16} & 7.92 & 689.59 \\
			& 100 & 20 & 10 & 110.57 & 17.96 & 0.16 & 95.65 & 18.20 & 2800.76 & \textbf{95.40} & 18.11 & 2714.30 \\
			\addlinespace[4pt]
			0.8 & 15  & 15 & 10 & 0.63 & 0.70 & 0.00 & \textbf{0.60} & 0.72 & 9.97 & \textbf{0.60} & 0.72 & 10.12 \\
			& 20  & 15 & 10 & 0.82 & 0.48 & 0.00 & \textbf{0.68} & 0.48 & 23.11 & \textbf{0.68} & 0.48 & 23.38 \\
			& 20  & 20 & 10 & 0.45 & 0.46 & 0.00 & \textbf{0.28} & 0.34 & 23.09 & \textbf{0.28} & 0.34 & 22.08 \\
			& 30  & 15 & 10 & 0.57 & 0.43 & 0.00 & 0.46 & 0.41 & 45.11 & \textbf{0.45} & 0.41 & 44.96 \\
			& 30  & 20 & 10 & 0.63 & 0.58 & 0.00 & \textbf{0.33} & 0.27 & 60.98 & 0.34 & 0.27 & 61.38 \\
			& 50  & 15 & 10 & 1.83 & 1.43 & 0.01 & \textbf{0.98} & 0.90 & 212.16 & \textbf{0.99} & 0.91 & 215.48 \\
			& 50  & 20 & 10 & 1.59 & 1.17 & 0.03 & 0.91 & 0.65 & 242.87 & \textbf{0.90} & 0.65 & 250.55 \\
			& 100 & 20 & 10 & 7.67 & 2.24 & 0.17 & \textbf{5.41} & 1.70 & 1592.02 & 5.42 & 1.69 & 1619.29 \\
			\midrule
			0.2 & \textit{All} & & 80 & 34.40 & 16.41 & 0.01 & 22.67 & 13.98 & 361.69 & \textbf{22.64} & 13.94 & 372.46 \\
			0.5 & \textit{All} & & 80 & 40.71 & 31.70 & 0.02 & 29.47 & 29.27 & 580.38 & \textbf{29.43} & 29.19 & 562.47 \\
			0.8 & \textit{All} & & 80 & 1.77 & 2.53 & 0.03 & \textbf{1.21} & 1.79 & 276.16 & \textbf{1.21} & 1.79 & 280.90 \\
			\bottomrule
		\end{tabular}
	\end{table}
	
	The variability of the obtained solutions is consistently higher for the constructive heuristic than for the two genetic algorithms across all conflict graph densities. Overall, the variability tends to decrease as the conflict graph becomes denser, although this trend is not systematic for the Taillard instances. For dense graphs ($p=0.8$), the standard deviations remain small, with values not exceeding $2.53\%$. These results indicate that the genetic algorithms provide more consistent solution quality, while stronger conflict restrictions generally reduce the variability of the obtained solutions.
	
	\begin{table}[H]
		\centering
		\small
		\setlength{\tabcolsep}{3.5pt}
		\caption{Results on generalized job shop instances.}
		\label{tab:main_Generalized}
		\begin{tabular}{c r r r rrr rrr rrr}
			\toprule
			& & & & \multicolumn{3}{c}{$H$} & \multicolumn{3}{c}{$GA$} & \multicolumn{3}{c}{$GA_h$} \\
			\cmidrule(lr){5-7} \cmidrule(lr){8-10} \cmidrule(lr){11-13}
			$p$ & $n$ & $m$ & \#Inst & RPD & $\sigma$ & Time & RPD & $\sigma$ & Time & RPD & $\sigma$ & Time \\
			\midrule
			0.2 & 10 & 5 & 15 & 40.97 & 28.36 & 0.00 & 30.97 & 24.98 & 2.90 & \textbf{30.78} & 24.96 & 2.88 \\
			& 20 & 10 & 15 & 32.75 & 13.66 & 0.00 & \textbf{25.94} & 14.36 & 11.03 & \textbf{25.95} & 14.57 & 10.51 \\
			& 50 & 15 & 15 & \textbf{61.88} & 22.43 & 0.00 & 62.32 & 21.73 & 75.48 & 62.33 & 21.91 & 81.23 \\
			& 100 & 20 & 15 & \textbf{103.98} & 30.88 & 0.07 & 106.58 & 29.62 & 98.43 & 106.56 & 29.47 & 99.53 \\
			\addlinespace[4pt]
			0.5 & 10 & 5 & 15 & 12.98 & 8.34 & 0.00 & 6.94 & 6.19 & 3.36 & \textbf{6.93} & 6.28 & 3.29 \\
			& 20 & 10 & 15 & 22.31 & 11.22 & 0.00 & \textbf{16.94} & 10.59 & 19.60 & 17.02 & 10.71 & 20.94 \\
			& 50 & 15 & 15 & 56.50 & 12.39 & 0.02 & \textbf{49.92} & 12.23 & 147.11 & 50.12 & 12.28 & 157.34 \\
			& 100 & 20 & 15 & 100.60 & 12.99 & 0.26 & \textbf{95.39} & 11.76 & 187.35 & 95.47 & 11.88 & 180.66 \\
			\addlinespace[4pt]
			0.8 & 10 & 5 & 15 & 2.84 & 2.37 & 0.00 & 1.81 & 1.46 & 1.89 & \textbf{1.79} & 1.45 & 1.97 \\
			& 20 & 10 & 15 & 1.66 & 1.18 & 0.00 & 0.94 & 0.74 & 6.11 & \textbf{0.92} & 0.75 & 6.33 \\
			& 50 & 15 & 15 & \textbf{6.28} & 3.60 & 0.03 & 7.42 & 3.64 & 37.41 & 7.42 & 3.70 & 47.46 \\
			& 100 & 20 & 15 & 28.88 & 6.88 & 0.31 & \textbf{24.81} & 6.34 & 470.20 & \textbf{24.81} & 6.31 & 457.63 \\
			\midrule
			0.2 & \textit{All} &  & 60 & 59.89 & 36.80 & 0.02 & 56.45 & 39.58 & 46.96 & \textbf{56.40} & 39.61 & 48.54 \\
			0.5 & \textit{All} &  & 60 & 48.10 & 36.39 & 0.07 & \textbf{42.30} & 36.29 & 89.35 & 42.38 & 36.34 & 90.56 \\
			0.8 & \textit{All} &  & 60 & 9.92 & 11.87 & 0.08 & 8.74 & 10.35 & 128.90 & \textbf{8.74} & 10.36 & 128.35 \\
			\bottomrule
		\end{tabular}
	\end{table}
	
	\begin{table}[htbp]
		\centering \small
		\caption{Friedman test and unadjusted pairwise Wilcoxon signed-rank $p$-values for makespan RPD across benchmarks and conflict-graph densities. Post-hoc conclusions use the Bonferroni-adjusted threshold $\alpha=0.0167$ within each benchmark--density group.}
		\label{tab:stat_tests}
		\begin{tabular}{ll rr ccc}
			\toprule
			Benchmark & $p$ & $\chi^2$ & $p_{\text{Fried.}}$ & $H$ vs.\ GA & $H$ vs.\ GA$_h$ & GA vs.\ GA$_h$ \\
			\midrule
			Lawrence-derived & 0.2 & 63.21 & 1.9e-14$^{***}$ & 1.8e-12$^{***}$ & 1.8e-12$^{***}$ & 0.4540 \\
			& 0.5 & 60.00 & 9.4e-14$^{***}$ & 1.8e-12$^{***}$ & 1.8e-12$^{***}$ & 0.8995 \\
			& 0.8 & 27.55 & 1.0e-06$^{***}$ & 6.1e-05$^{***}$ & 6.1e-05$^{***}$ & -- \\
			\addlinespace[3pt]
			Taillard-derived & 0.2 & 122.79 & 2.2e-27$^{***}$ & 7.8e-15$^{***}$ & 7.8e-15$^{***}$ & 0.2765 \\
			& 0.5 & 118.90 & 1.5e-26$^{***}$ & 1.1e-14$^{***}$ & 1.1e-14$^{***}$ & 0.5617 \\
			& 0.8 & 88.43 & 6.3e-20$^{***}$ & 1.6e-10$^{***}$ & 1.1e-10$^{***}$ & 0.6225 \\
			\addlinespace[3pt]
			Generalized job shop & 0.2 & 5.33 & 0.0695 & 0.0040 & 0.0037 & 0.3200 \\
			& 0.5 & 75.16 & 4.8e-17$^{***}$ & 2.8e-11$^{***}$ & 3.3e-11$^{***}$ & 0.0731 \\
			& 0.8 & 20.92 & 2.9e-05$^{***}$ & 3.7e-04$^{***}$ & 3.4e-04$^{***}$ & 0.3104 \\
			\midrule
			\textit{All benchmarks} & 0.2 & 154.41 & 3.0e-34$^{***}$ & 1.0e-25$^{***}$ & 9.0e-26$^{***}$ & 0.3158 \\
			& 0.5 & 252.26 & 1.7e-55$^{***}$ & 4.9e-31$^{***}$ & 5.1e-31$^{***}$ & 0.5147 \\
			& 0.8 & 121.39 & 4.4e-27$^{***}$ & 1.3e-12$^{***}$ & 9.9e-13$^{***}$ & 0.5704 \\
			\bottomrule
		\end{tabular}
	\end{table}

	\subsubsection{Solution quality distribution}
	
	Figure~\ref{fig:boxplots} displays the RPD distribution for each
	method across all benchmarks, separated by density.  The boxes
	span the interquartile range (IQR); whiskers extend to
	1.5$\times$IQR; points beyond the whiskers are outliers.
	
	\begin{figure}[htbp]
		\centering
		\includegraphics[width=\textwidth]{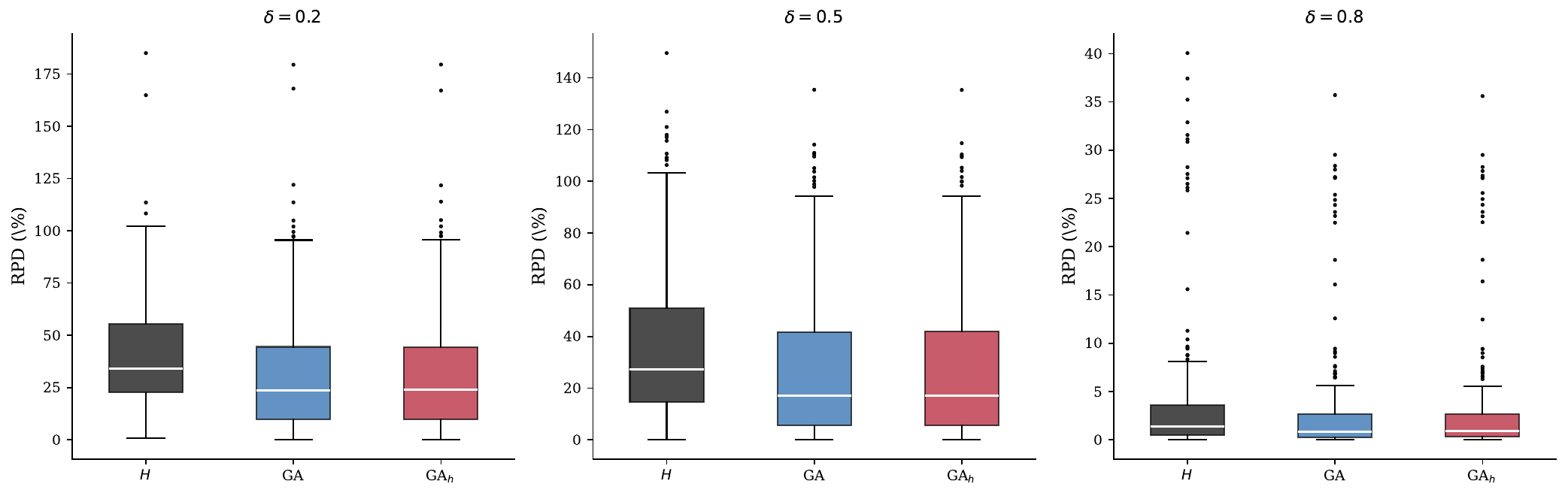}
		\caption{Box plots of RPD (\%) by method and density across all benchmarks.}
		\label{fig:boxplots}
	\end{figure}

	The box plots reveal clear distributional differences. At $p=0.2$, $H$ has both a higher median and a wider interquartile range than $GA$ and $GA_h$, indicating greater sensitivity of the constructive heuristic to instance structure. The metaheuristic variants produce tighter distributions, suggesting that the evolutionary search has a stabilising effect. At $p=0.8$, all three boxes collapse to a narrow band near zero, with very few outliers, consistent with the constrained nature of dense conflict graphs. The outliers visible at low and medium densities typically correspond to the largest instances (e.g., $100\times20$), where the gap to the lower bound is inherently large.
	
	\subsubsection{Average RPD by benchmark}
	
	\begin{figure}[htbp]
		\centering
		\includegraphics[width=\textwidth]{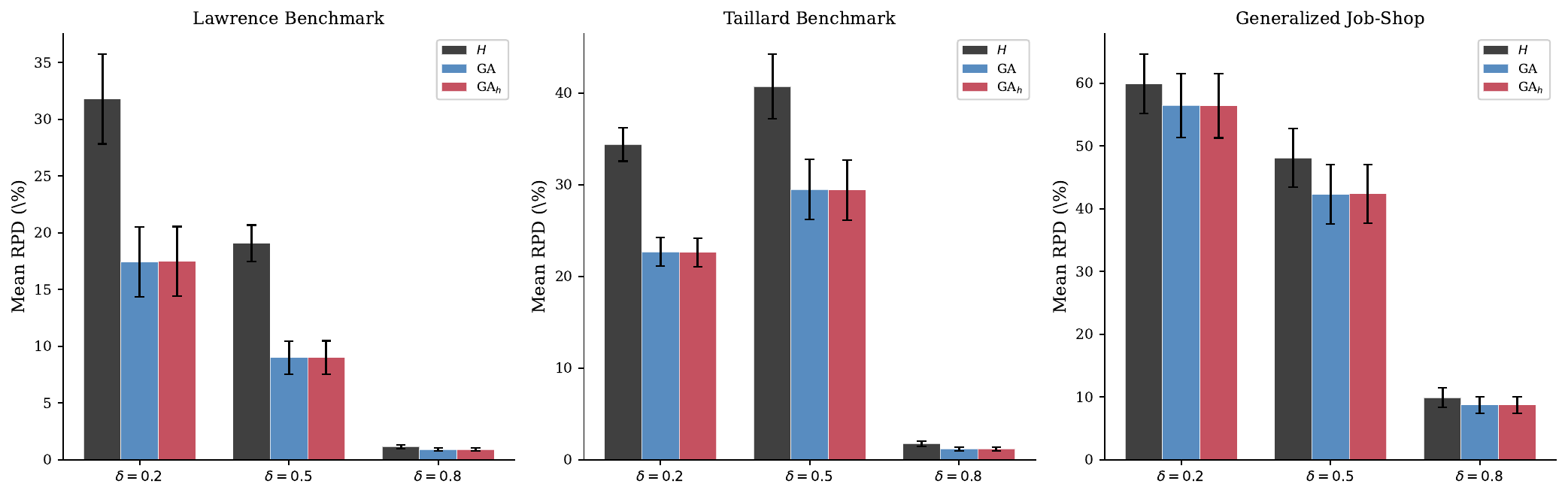}
		\caption{Mean RPD (\%) with standard-error bars by benchmark and density.}
		\label{fig:barchart}
	\end{figure}
	
	Figure~\ref{fig:barchart} decomposes the average RPD by benchmark family and density, providing a direct visual comparison of the relative difficulty of each benchmark. The Generalized job shop instances consistently yield the tallest bars, confirming their greater difficulty across all density levels, while the Lawrence instances are the easiest. Within each benchmark group, $H$ bars are systematically taller than the $GA$ and $GA_h$ bars at $p=0.2$ and $p=0.5$. At $p=0.8$, the three bars are nearly equal, consistent with the convergence observed in the tables. The standard-error bars are visibly wider for $H$ than for the GA variants, reinforcing the conclusion that metaheuristics produce more predictable solution quality.
	
	\section{Conclusions}\label{sec:conclusion}
	
	In this paper, we studied the job shop scheduling problem with
	conflict graph constraints, with the objective of minimizing the
	makespan. We established its polynomial equivalence to a
	resource-constrained job shop problem and investigated the
	impact of the conflict graph structure on its computational
	complexity. In particular, we derived an explicit formula for
	the optimal makespan and a linear-time algorithm for the
	two-machine case with unit-time operations and a cluster
	conflict graph, without restricting the number of cliques
	or operations per job. This result also establishes the
	polynomial solvability of
	$J2|res^t111, p_{ij}=1|C_{\max}$.
	Furthermore, we developed a polynomial-time
	$4/3$-approximation algorithm for arbitrary conflict graphs
	when each job consists of exactly two unit-time operations. We also established improved approximation
	guarantees under additional restrictions on the agreement
	graph. For the general problem, we proposed four mixed-integer
	linear programming formulations, lower bounds, and a genetic
	algorithm incorporating different schedule generation
	strategies. Computational experiments showed that the
	precedence-based formulations generally outperform the
	time-indexed formulations, while the agreement graph-based
	lower bounds are particularly effective for medium and dense
	conflict graphs. The genetic algorithm generally outperformed
	the constructive heuristics, with the hybrid schedule
	evaluation strategy providing a competitive balance between
	solution quality and computational effort. Future research could investigate additional graph classes
	admitting polynomial-time algorithms or improved approximation
	guarantees. Another direction is the development of more
	effective exact methods, including column-generation-based
	approaches~\cite{tellache2024linear} and branch-and-cut
	algorithms~\cite{tellache2025variant}, as well as extensions
	to online scheduling.
	
	\bibliographystyle{unsrtnat}
	\bibliography{biblio}
	
\end{document}